\documentclass[10pt, aps, english, prx, superscriptaddress, twocolumn,
notitlepage, footinbib, longbibliography, 
floatfix]{revtex4-2}
\usepackage[utf8]{inputenc}
\usepackage[english]{babel}
\usepackage{natbib}
\usepackage{graphicx}
\usepackage{mathtools}
\usepackage{amsfonts} 
\usepackage{amsmath} 
\usepackage{amssymb}
\usepackage{amsthm}
\usepackage{mathrsfs}
\usepackage{array}
\usepackage{makecell}
\usepackage{bm}
\usepackage{braket}
\usepackage{color} 
\usepackage{bbold}
\usepackage{hyperref}
\usepackage{import}
\usepackage{microtype}
\usepackage{relsize}
\usepackage{framed}
\usepackage[normalem]{ulem} 
\usepackage[dvipsnames]{xcolor}
\usepackage[percent]{overpic}
\usepackage{tikz}
\usetikzlibrary{calc, arrows.meta}
\usepackage{pgfplots}
\pgfplotsset{compat=1.18}
\usepackage{enumerate}
\usepackage{mdframed}

\usepackage{ifthen}
\usepackage{fontawesome5}

\newcommand{\drawplot}[4]{%
    \begin{tikzpicture}
        \draw[->, thick] (0,0) -- (2,0) node[right] {$\beta$};
        \draw[->, thick] (0,0) -- (0,1.5) node[left] {#1};

        \ifthenelse{\equal{#2}{SDOE}}{ 
            \draw[red, very thick] (0,0) -- (0.7,0) .. controls (0.7,1) and (1.5,1.1) .. (2,1.2); 
            \draw (0.35,0.35) node {\faSkull};
        }{
            \ifthenelse{\equal{#2}{NoSDOE}}{ 
                \draw[blue, very thick] (0,0) .. controls (0.5,1) and (1.5,1.1) .. (2,1.2) ;
            }{}
        }

        \ifthenelse{\equal{#3}{below}}{
            \draw (1,-0.1) node[anchor=north] {#4};
        }{\ifthenelse{\equal{#3}{above}}{ 
            \draw (1,1.6) node[anchor=south] {#4};    
        }{\ifthenelse{\equal{#3}{center}}{
            \draw (0.65, 0.45) node[anchor=west, fill=blue!10, rounded corners] {\scalebox{0.7}{#4}};
        }{}}}
    \end{tikzpicture}
    
}

\def\app#1#2{%
  \mathrel{%
    \setbox0=\hbox{$#1\sim$}%
    \setbox2=\hbox{%
      \rlap{\hbox{$#1\propto$}}%
      \lower1.1\ht0\box0%
    }%
    \raise0.25\ht2\box2%
  }%
}
\def\approxprop{\mathpalette\app\relax}

\newcommand{\norm}[1]{\left\lVert#1\right\rVert}
\newcommand{\defeq}{\vcentcolon=}

\newcommand{\ketbra}[2]{\mathinner{\ket{#1}\!\bra{#2}}}

\newcommand{\one}{\mathbb{1}}

\newcommand{\Z}{\mathbb{Z}}
\newcommand{\N}{\mathbb{N}}
\newcommand{\C}{\mathbb{C}}

\newcommand{\Hilb}{\mathcal{H}}

\DeclareMathOperator{\Tr}{Tr}

\DeclareMathOperator{\supp}{supp}

\DeclareMathOperator{\diam}{diam}

\DeclareMathOperator{\diag}{diag}

\newtheorem{theorem}{Theorem}

\newtheorem{proposition}{Proposition}

\newtheorem*{corollary*}{Corollary}
\newtheorem{lemma}{Lemma}

\theoremstyle{definition}
\newtheorem{definition}{Definition}

\mdfdefinestyle{resultstyle}{%
  innertopmargin=0pt,
  innerbottommargin=6pt,
  innerleftmargin=6pt,
  innerrightmargin=6pt,
  skipabove=\baselineskip,
  skipbelow=\baselineskip,
}
\newmdtheoremenv[style=resultstyle]{result}{Result}

\newcommand{\SEP}{\textsf{SEP}}
\newcommand{\SymSEP}{\textsf{SymSEP}}

\newcommand{\PPT}{\textsf{PPT}}
\newcommand{\SDOE}{\textsf{SDOE}}
\newcommand{\POE}{\textsf{POE}}
\newcommand{\SEC}{\textsf{SEC}}
\newcommand{\EC}{\textsf{EC}}
\newcommand{\NC}{\textsf{NC}}

\begin{document}
\title{Symmetry enforces entanglement at high temperatures}

\author{Amir-Reza Negari}

\thanks{These authors contributed equally to this work.}
\affiliation{Perimeter Institute for Theoretical Physics, Waterloo, Ontario N2L 2Y5, Canada}
\affiliation{Department of Physics and Astronomy, University of Waterloo, Waterloo, Ontario N2L 3G1, Canada}

\author{Leonardo A. Lessa}

\thanks{These authors contributed equally to this work.}
\affiliation{Perimeter Institute for Theoretical Physics, Waterloo, Ontario N2L 2Y5, Canada}
\affiliation{Department of Physics and Astronomy, University of Waterloo, Waterloo, Ontario N2L 3G1, Canada}

\author{Subhayan Sahu}
\email{ssahu@perimeterinstitute.ca}
\affiliation{Perimeter Institute for Theoretical Physics, Waterloo, Ontario N2L 2Y5, Canada}

\begin{abstract}
    Many-body quantum systems with local interactions undergo ``sudden death of entanglement'' at high temperatures, whereby thermal states become classical mixtures of product states. We investigate whether symmetry constraints can prevent this phenomenon. We prove that strongly symmetric thermal states (canonical ensemble) of generic Hamiltonians with on-site Abelian symmetries remain entangled with non-zero entanglement negativity at arbitrarily high temperatures, under mild conditions on the symmetry actions and the charge sector of the strong symmetry. Our results extend to weakly symmetric thermal states (Gibbs ensemble) under superselection rules, which restrict state decompositions to be symmetric. In particular, we show that fermionic Gibbs states evade sudden death of entanglement and have persistent fermionic negativity at high temperatures, proving along the way some existing conjectures about fermionic entanglement. These findings demonstrate that global symmetry correlations can preserve quantum entanglement despite thermal decoherence, providing new insights into the interplay between symmetry and quantum information in thermal equilibrium. 
\end{abstract}
\maketitle

Many-body systems equilibrate to their Gibbs distribution at finite temperatures by exchanging energy and information with a bath. The resulting state $\rho_\beta \propto e^{-\beta H}$ for system Hamiltonian $H$ at inverse temperature $\beta = T^{-1}$ is a mixed state that maximizes entropy under energy constraints, reflecting an observer's ignorance of information leaked to the bath. For quantum Hamiltonians with local interactions, the Gibbs state has structural properties at finite non-zero temperatures that reflect its low complexity, such as the area law of mutual information~\cite{Wolf_2008}, and the local Markov property~\cite{chen2025quantumgibbsstateslocally, kuwahara2024clusteringconditionalmutualinformation, Kato_2019}. One such property that is manifestly non-classical is entanglement, which was recently proven \cite{bakshi2025hightemperaturegibbsstatesunentangled} to undergo ``sudden death'' at sufficiently high but finite temperatures in the thermodynamic limit, thus retaining only classical correlations. This result confirmed evidence from previously studied few-body and solvable models~\cite{Arnesen_2001,Fine_2005,Qasimi2008,Yu_2009, Aolita_2015, Sherman_2016, Hart_2018, parez2025fateentanglement, gurvits_separable_2003}.

\begin{table}[t]
    \centering
    \begin{tabular}{c|c|c}
         & No superselection & \makecell{Superselection \\ (e.g. fermions)} \\ \hline
        \raisebox{0.8cm}{\makecell{Gibbs \\ ensemble \\ $e^{-\beta H}$}} & \drawplot{$E$}{SDOE}{}{} & \drawplot{$E_S$}{NoSDOE}{center}{$[U_A(g), H] \neq 0$\ (\hyperref[eq:symm_entangling_cond]{\SEC})} \\ \hline 
        \raisebox{0.8cm}{\makecell{Canonical \\ ensemble \\ $e^{-\beta H} \Pi_\Lambda$}} & \multicolumn{2}{c}{\drawplot{$E$}{NoSDOE}{center}{$[U_A(g), H] \Pi_\Lambda \neq 0$\ (\hyperref[eq:entangling_cond]{\EC})}}
    \end{tabular}
    \caption{Typical behavior of entanglement measures ($E$, $E_S$) as functions of the inverse temperature $\beta = T^{-1}$ for the Gibbs and canonical ensembles in the cases with and without superselection. $E$ ($E_S$) is a measure of (symmetric) entanglement, such as (fermionic) negativity. There is sudden death of entanglement for the Gibbs ensemble with no superselection rule if $H$ is local \cite{bakshi2025hightemperaturegibbsstatesunentangled}, while there is not for the other cases if their respective entangling conditions (\EC\ and \SEC) are satisfied. For an example of the above, see Fig. \ref{fig:cluster_chain_plot}.}
    \label{tab:schematic_plots}
\end{table}

A natural question arises: can constraints allow us to evade such sudden death of entanglement? Recent developments in mixed-state phases of matter have confirmed the essential role of symmetries \cite{lee_quantum_2023, sala_spontaneous_2024, lessa_strong--weak_2024, zhang_strong--weak_2024, ma_average_2023, ma_topological_2025, lee_symmetry_2025, xue_tensor_2024, ma_symmetry_2024} and uncovered novel connections to their global entanglement properties~\cite{Moharramipour_2024, Li_2025, sahu2025entanglementcosthierarchiesquantum, Lessa_2025MixedStateQuantumAnomaly, lessa_higher-form_2025, li_how_2024, wang_intrinsic_2025, li_symmetry-enforced_2024, wang_anomaly_2024, hsin_higher-form_2025}. In this Letter, we prove that symmetry restrictions on thermal states generically prevent sudden death of entanglement (\SDOE) at all finite temperatures (See Table \ref{tab:schematic_plots}).

Specifically, we consider a symmetry group $G$, and a Hamiltonian $H$ that is symmetric under $G$: $\forall g \in G, [U(g), H] = 0$, where $U = \bigotimes_i u_i$ is an on-site symmetry representation of $G$ in the total Hilbert space $\Hilb$. Gibbs states of local Hamiltonians $H$ that are symmetric also undergo \SDOE, so we must impose a stronger symmetry constraint on the Gibbs state to see meaningful effects on entanglement. We consider two physically motivated ways of doing that:

One way is to impose \textbf{strong symmetry} in the thermal ensemble $\rho_\beta = \sum_i p_i \ketbra{\psi_i}{\psi_i}$ by limiting the mixture only to states $\ket{\psi}$
that have a fixed symmetry charge $\Lambda : G \to U(1)$, i.e. $U(g) \ket{\psi} = \Lambda(g) \ket{\psi}$. This constraint results in the \textit{canonical ensemble} $\rho_{\beta, \Lambda} \propto e^{-\beta H} \Pi_{\Lambda}$, where $\Pi_\Lambda$ projects onto the charge sector $\Lambda$. By definition, $\rho_{\beta, \Lambda}$ is \textit{strongly symmetric} under $G$, satisfying $U(g) \rho_{\beta, \Lambda} = \Lambda(g) \rho_{\beta, \Lambda}$, which is a stronger requirement than the usual \textit{weak symmetry} condition $U(g) \rho U(g)^\dagger = \rho$ \cite{Buca_2012, Albert_2014, albert_lindbladians_2018}.

Alternatively, we can impose a \textbf{superselection rule} by considering a restricted quantum theory consisting only of states that are (weakly) symmetric under $G$ \cite{bartlett_reference_2007, giulini_superselection_2009, centeno_twirled_2025}. This restriction can either be fundamental to physical reality, as in the case of the fermion parity superselection rule \cite{hegerfeldt_proof_1968}, or be put by hand, as in real quantum theory \cite{stueckelberg_quantum_1960, caves_entanglement_2001, ying_whether_2025}. Crucially, the superselection rule alters the definition of entanglement itself \cite{Verstraete_2003, Schuch_2004nonlocal, Schuch_2004quantumentanglement, Banuls_2007, dariano_feynman_2014, negari2025extendibilityfermionicgaussianstates, baldijao_tomographicallynonlocal_2026}. Namely, a bipartite state $\rho_{AB}$ is \textit{symmetrically separable} if it is the mixture of tensor product states $\rho_{A,i} \otimes \rho_{B,i}$ whose parts, $\rho_{A,i}$ and $\rho_{B,i}$, are weakly symmetric themselves, thus valid states within the superselected theory. This means that $\rho_{AB}$ is not only (weakly) symmetric and separable, but also \textit{locally symmetric}, since $[U_A(g), \rho_{AB}] = \sum_i p_i [U_A(g), \rho_{A,i}] \otimes \rho_{B,i} = 0$, and similarly for $U_B$. The converse is also true, since a locally symmetric ensemble of product states $\rho_{A,i} \otimes \rho_{B,i}$ is also an ensemble of the symmetrized states $\tau_A(\rho_{A,i}) \otimes \tau_B(\rho_{B,i})$, and thus symmetrically separable, where $\tau_{A,B}(X) \defeq \frac{1}{|G|} \sum_{g \in G} U_{A,B}(g) X U_{A,B}(g)^{\dagger} $ is the local twirling channel.

The core result of the paper is as follows: since the local symmetry condition $[U_{A,B}(g), \rho_{\beta}] = 0$ is an additional linear constraint that a separable state $\rho_{\beta<\beta_{\SDOE}}$ must satisfy to be symmetrically separable, it can be easily violated, unless $\rho_\beta$ is fine-tuned. More precisely,
\begin{result}\label{result:POE_superselection}
    For Abelian on-site symmetries and \textit{generic} symmetric local Hamiltonians, the Gibbs ensemble $\rho_\beta$ is symmetrically entangled at sufficiently high temperatures.
\end{result}

Although strong symmetries and superselection rules constrain thermal states in different ways, we prove that their entanglement properties are tightly connected in the following sense: the Gibbs ensemble $\rho_{\beta}$ is symmetrically entangled (i.e. \textit{not} symmetrically separable) if and only if the canonical ensemble $\rho_{\beta, \Lambda}$ is also entangled for some charge $\Lambda$. Then, we can prove a similar result for the persistence of entanglement~\footnote{Not to be confused with the robustness measure of multipartite entanglement called ``persistency of entanglement''~\cite{briegel_persistent_2001, Brunner_2012}.} of the canonical ensemble:

\begin{result}\label{result:POE_strong}
    Under the same conditions of Result \ref{result:POE_superselection}, the canonical ensemble $\rho_{\beta, \Lambda}$ is entangled with non-zero entanglement negativity, for almost all symmetry actions and irreps $\Lambda$ in the thermodynamic limit.
\end{result}

Result \ref{result:POE_strong} is stronger than Result \ref{result:POE_superselection} in two ways: first, we provide necessary and sufficient conditions for $\rho_{\beta, \Lambda}$ to be entangled, with $\Lambda$ arbitrarily chosen; and second, $\rho_{\beta, \Lambda}$ is not only entangled but has nonzero negativity, which is an efficiently computable measure of entanglement \cite{Peres_1996, Horodecki_1996,vidal_computable_2002}.

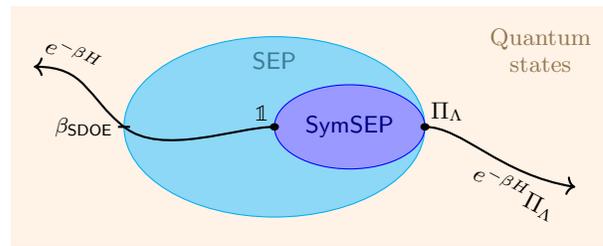
\begin{figure}[t]
    \centering





\begin{tikzpicture}
    \begin{scope}[yscale=0.8]
    \fill[orange!10] (-4,-2) rectangle (4,2); 
    \draw[YellowOrange!50!black] (4,2) node[below left, align=center, inner sep=8pt] {Quantum\\states};

    \draw[draw=gray, dashed, very thick] (-4,0) -- (4,0);
    \draw[gray] (3.1, -0.3) node {\footnotesize $[U_A, \rho] = 0$};    

    \begin{scope}[shift={(-0.5,0)}]
    \draw[fill=green!30!white, draw=green] (0,0) ellipse [x radius = 2, y radius = 1.5];
    \draw[green!60!black] (0, 1.1) node {\SEP};

    \draw[draw=green!50!black, very thick] (-2,0) -- (2,0);
    \draw[green!30!black] (1.13, 0.25) node {\footnotesize \SymSEP};

    \draw[fill=black] (0,0) circle [radius=1.5pt] node[above left, inner sep=2pt] {$\one$};
    \draw[fill=black] (2,0) circle [radius=1.5pt] node[above right, inner sep=2pt] {$\Pi_\Lambda$};
    \draw[black, thick] (-1.85,0.49) -- (-1.97,0.41) node[left, inner sep=1.5pt] {\footnotesize $\beta_\SDOE$};

    \draw[blue, ->, thick] (2,0) .. controls (2.5, -0.3) and (2.5, -0.8) .. (4, -1.5) node[near end, sloped, below] {\small $e^{-\beta H} \Pi_\Lambda$};
    \draw[red, ->, thick] (0,0) .. controls (-2.5, 0.3) and (-2.5, 1) .. (-3.2, 1) node[near end, sloped, above] {\small $e^{-\beta H}$};
    \end{scope}
    \end{scope}

\end{tikzpicture}
    \caption{Schematic paths of generic symmetric many-body thermal states in relation to the set of (symmetrically) separable states \SEP\ (\SymSEP). The Gibbs ensemble $\rho_\beta \propto e^{-\beta H}$ is separable from $\beta = 0$ to $\beta_{\SDOE}$ due to the sudden death of entanglement, but it becomes symmetrically entangled as soon as $\beta > 0$. Similarly, the canonical ensemble $\rho_{\beta, \Lambda} \propto e^{-\beta H} \Pi_\Lambda$ is also entangled for $\beta > 0$. As explained in the main text, $\SymSEP$ is the intersection of $\SEP$ with the set of locally symmetric states, i.e. those satisfying $[U_A(g), \rho] = 0$ for all $g \in G$.     }
    \label{fig:geometry_schematic}
\end{figure}

Non-Abelian symmetries~\cite{Moharramipour_2024,Li_2025}, quantum Hilbert-space fragmentation ~\cite{sahu2025entanglementcosthierarchiesquantum}, and anomalous non-onsite symmetries \cite{Lessa_2025MixedStateQuantumAnomaly, li_how_2024, wang_intrinsic_2025, li_symmetry-enforced_2024, wang_anomaly_2024, hsin_higher-form_2025} have previously been found to guarantee entanglement at infinite temperature. However, for on-site Abelian symmetries, the canonical Gibbs states are fully symmetrically separable at infinite temperature: our results prove that entanglement, however, generically persists at all finite temperatures.

Recent works report similar findings on the persistence of entanglement~\cite{garratt2025entanglementprivateinformationmanybody} and negativity~\cite{kim2025persistenttopologicalnegativityhightemperature} in canonical ensembles at high temperatures. Our work provides significant generalization through distinct proof techniques and complete characterization of symmetries and irreps with persistent entanglement. 

Furthermore, we unify the ``no \SDOE'' result for canonical ensembles with the persistence of symmetric entanglement in thermal states under superselection rule, which were previously noted for fermionic systems \cite{Banuls_2007, Benatti_2014, Shapourian_2019finite_temperature, Choi_2024, ramkumar2025hightemperaturefermionicgibbsstates}. More specifically for fermions, we are also able to show an equivalent version of Result \ref{result:POE_strong} where the entanglement negativity measure used is the appropriate one for fermionic systems \cite{Shapourian_2017, Shapourian_2019}, along the way proving existing conjectures in~\cite{Shapourian_2019, Shapourian_2019finite_temperature}.

\textit{\textbf{Setup and notation---}}
We consider a multipartite Hilbert space with $N < \infty$ identical sites, $\Hilb = \bigotimes_{i=1}^N \Hilb_i$, $\forall i,j.\Hilb_i \simeq \Hilb_j$. The symmetry group $G$ is Abelian and acts as a on-site unitary $U : G \to \mathcal{L}(\mathcal{H})$, $U = \bigotimes_i u_i$ in a non-trivial way on each site, i.e. $\exists g \in G, u(g) \not\propto \one$. Sometimes, we will specialize to homogeneous symmetry actions $U = u^{\otimes N}$.

From the symmetry action of $G$ on the Hilbert space $\Hilb$, it is decomposed into $\Hilb = \bigoplus_\lambda V_{\lambda}^{\oplus m_\lambda}$, where the direct sum is over all irreps $\lambda$ of $G$, with $m_\lambda$ being the multiplicities. Since $G$ is Abelian, each irrep space $V_\lambda$ is one-dimensional, and so the irreps are identified with their characters, $\lambda : G \to U(1)$. By Schur's lemma, a weakly symmetric state $\rho$ can be decomposed into strongly symmetric block diagonal components, $\rho = \sum_{\lambda} q_{\lambda} \rho_{\lambda}$, where $\rho_{\lambda} = \rho \Pi_{\lambda} / q_\lambda$ with $\Pi_\lambda = \sum_{g \in G} \overline{\lambda(g)} U(g)$ being the projector onto the \textit{isotypic component} $V_\lambda^{\oplus m_\lambda}$, and $q_\lambda = \Tr[\rho \Pi_{\lambda}]$. We shall call $\sum_{\lambda} q_\lambda \rho_\lambda$ the \textit{isotypic decomposition} of $\rho$.

\textit{\textbf{Persistence of entanglement---}}
Now consider the Gibbs state $\rho_{\beta} \propto e^{-\beta H}$. If $H$ is local, $\rho_\beta$ becomes fully separable at high temperatures. We will now give conditions under which not only symmetry constraints \textit{prevent} \SDOE, but also leads to the \textit{persistence} of entanglement (\POE) at high temperature, by which we mean that there exists a $\beta^* > 0$ such that $\rho_\beta$ is symmetrically entangled if $\beta < \beta^*$, with an analogous statement for $\rho_{\beta, \Lambda}$.

We proceed by assuming that the entanglement of $\rho_{\beta}$ does \textit{not} persist. That is equivalent to $\rho_\beta$ being symmetrically separable for all $\beta$ in a decreasing sequence $(\beta_n)_{n \in \N}$ converging to zero: $\rho_{\beta} = \sum_{i}p_i \rho_{\beta,A,i}\otimes \rho_{\beta, B,i}$ for a bipartition $A|B$. Since the $\rho_{\beta, A,i}$ states are individually symmetric, we have $\forall g \in G, [U_A(g), \rho_{\beta}] = \sum_{i}p_i[U_A(g), \rho_{\beta, A,i}]\otimes \rho_{\beta, B,i} = 0$. Since each matrix element of $[U_A(g), \rho_{\beta}]$ is a holomorphic function of $\beta$ that is zero on the set $\{\beta_n \}$ with the accumulation point $\beta =0$, it is identically zero for all $\beta \in \C$. In particular, its derivative at $\beta = 0$ must also be zero, which leads to a necessary condition for symmetric separability at high temperatures: $\forall g\in G,\ [U_A(g), H] = 0$.

By the converse of the reasoning above, we arrive at a sufficient condition for the Gibbs state to be symmetrically entangled at arbitrarily high temperatures $\beta \in (0, \beta^*)$,
\begin{align}\label{eq:symm_entangling_cond}
    \forall g\in G,\  &[U_A(g), H] \neq 0, \nonumber \\ &~~\text{[Symmetric entangling condition (\SEC)]}
\end{align}

By the correspondence between the symmetric separability of $\rho_\beta$ and the separability of $\rho_{\beta, \Lambda}$, \textsf{SEC} also implies the existence of at least one irrep $\Lambda$ whose canonical state exhibits \POE. Beyond this, following the same argument above for $\rho_{\beta, \Lambda} \propto e^{-\beta H} \Pi_\Lambda$, we can reach a different sufficient condition for \POE\ for \textit{any given $\Lambda$}, 
\begin{align}\label{eq:entangling_cond}
    \forall g\in G,\  [U_A(g), H]&~\Pi_{\Lambda} \neq 0. \nonumber \\ &~~\text{[Entangling condition (\EC)]}
\end{align}

Combining these results, we arrive at the following theorem, which is summarized in Table \ref{tab:schematic_plots}

\begin{theorem}[Persistence of entanglement]\label{thm:poe}
    Given an on-site Abelian symmetry,
    \begin{enumerate}
        \item For a symmetric Hamiltonian $H$ that satisfies \SEC, the Gibbs ensemble is symmetrically entangled, and $\exists$ an irrep $\Lambda$ such that the canonical ensemble is entangled at arbitrarily high temperatures.
        \item For a symmetric Hamiltonian $H$ and an irrep $\Lambda$ that satisfies \EC, the canonical ensemble is entangled at arbitrarily high temperatures.
    \end{enumerate}
\end{theorem}

\textit{\textbf{Genericness of entangling condition ---}} 

We now intend to prove that \SEC\ and \EC\ are generically satisfied. The crucial property of these conditions is that they are linear in $H$; for any \( H \) that violates $\SEC / \EC$, adding a perturbation \( V \) that satisfies $\SEC / \EC$ ensures that \( H + \epsilon V \) also satisfies for all \( \epsilon \neq 0 \). Thus, proving the existence of entangling perturbations suffices to prove that fine-tuning $H$ is required to violate $\SEC / \EC$.  

For example, consider a qubit chain with $ G =\mathbb{Z}_2$ symmetry generated by \( \bigotimes_i Z_i \). A fine-tuned Hamiltonian $H_Z$ consisting solely of \( Z \) terms yields separable thermal states $\rho_\beta$ and $\rho_{\beta, \lambda}$. However, the operators \( X_i \) and \( Y_i \) carry charge \(-1\), so \( V_{ij} = X_i X_j \) or \( Y_i Y_j \) is an entangling perturbation when \( i \in A \) and \( j \in B \), making $H_Z + \epsilon V_{ij}$ satisfy both \EC\ and \SEC, even for infinitesimal $\epsilon$.

For other symmetry groups, an entangling $V$ that satisfies \SEC\ can be similarly constructed by pairing a charged operator in \(A\) with an oppositely charged one in \(B\), so that \( [U(g), V] = 0 \) but \( [U_A(g), V] \neq 0 \). The \EC\, is more difficult to establish in general, due to the projector $\Pi_{\Lambda}$. For \textit{finite} Abelian symmetries with homogeneous symmetry action, $U = u^{\otimes N}$, we can invoke local indistinguishability of symmetric subspaces to justify the existence of $V$ satisfying both \EC\ and \SEC. Consider a large tripartition $A|B|C$, where $B$ is a buffer region that separates $A$ and $C$, taken to be large enough such that there are no local Hamiltonian terms in $H$ that is supported on $AC$. Since $\Pi_\Lambda$ is maximally mixed with just a \textit{global} charge constraint, we expect $\Tr_{C}\Pi_{\Lambda} \approxprop  \one_{AB}$, which we can quantitatively establish for finite Abelian groups. This implies that tracing out a sufficiently large region disjoint from the support of $V$ and $A$ reduces \( [U_A(g), V] \Pi_\Lambda = 0 \) to \( [U_A(g), V] \approx 0 \), with equality in the thermodynamic limit. In fact, for finite Abelian groups, \SEC\ and \EC\ are equivalent for large enough regions in the thermodynamic limit (See \cite{supp}).

More rigorously, we have a classification of irreps $\Lambda$ for which two-body entangling perturbations satisfying both \EC\ and \SEC exist. We first define the notion of semiuniformity of irreps: a global irrep $\Lambda$ of a $N$-partite system is \textbf{uniform} if its projector $\Pi_\lambda$ is a uniform tensor product of an on-site irrep $\lambda$: $\Pi_\Lambda = \bigotimes_{i=1}^N \Pi_\lambda^{(i)}$. $\Lambda$ is \textbf{semiuniform} if its projector $\Pi_\Lambda$ is a direct sum of uniform tensor products: $\Pi_\Lambda = \sum_\alpha \bigotimes_{i=1}^N \Pi_{\lambda_\alpha}^{(i)}$. If $G$ is finite, there are no semiuniform global irreps for $N > |G| = O(1)$; and as we comment on later, semiuniform global irreps are also quantitatively rare for continuous groups represented on $N\gg 1$ qudits. 

By excluding semiuniform irreps, we can construct two-body entangling Hamiltonians:
\begin{theorem}\label{thm:no_semiuniform-entangling_ham}
    Consider a multipartite system with homogeneous symmetry action $U$ and an irrep $\Lambda$ that is not semiuniform. Then, for any pair of sites $(i,j)$, there exists a two-body symmetric Hamiltonian $v_{i,j}$ supported on $\{i, j\}$ that entangles any bipartition that separates site $i$ from site $j$. Moreover, given a connected interaction graph, the uniform nearest-neighbor perturbation $\sum_{\langle i, j \rangle} v_{i,j}$ entangles all bipartitions.
    
\end{theorem}

We now sketch the construction of the entangling $v_{i,j}$, leaving a full proof to the Supplemental Materials (S.M.)\cite{supp}. For non-semiuniform irreps and for any pair of sites $i,j$ there exist at least two local irreps $\lambda' \neq \lambda$. Then, let $\sigma_i^+ \defeq \!\ketbra{\lambda'}{\lambda}_i$ and $\sigma_j^+ \defeq \!\ketbra{\lambda'}{\lambda}_j$. Denoting their Hermitian conjugates by $\sigma_{i,j}^- \defeq (\sigma_{i,j}^+)^\dagger$, we choose $v_{i,j} \defeq \sigma_i^+ \sigma_j^- + \sigma_i^- \sigma_j^+$. It is easy to check that $v_{i,j}$ is symmetric and satisfies the \EC\ if $i \in A$, since it pairs a charged operator in $A$ with an oppositely charged one in $B$. 

\textit{\textbf{Continuous symmetries ---}}For continuous groups such as $U(1)$, the situation can be more subtle. Consider the $U(1)$ action on a system of $N$ qubits as $\theta \mapsto e^{i \theta \sum_i Z_i}$. The construction of the two-body Hamiltonian term sketched in theorem~\ref{thm:no_semiuniform-entangling_ham} works for all irreps except for the maximal or the minimal charge irrep, with $\langle \sum_{i}Z_i \rangle = \pm N$. In fact, these sectors only consist of product states $\ket{\uparrow
\uparrow\cdots \uparrow}$ and $\ket{\downarrow\downarrow\cdots \downarrow}$ respectively, which always leads to an unentangled canonical ensemble, by negating \EC\ regardless of $V$. These extremal irreps are uniform, since their projectors are separable: $\otimes_{i}\ket{\uparrow}\bra{\uparrow}_i$ and $\otimes_{i}\ket{\downarrow}\bra{\downarrow}_i$, and can thus evade theorem~\ref{thm:no_semiuniform-entangling_ham}. However, the other charge sectors are not semiuniform and thus theorem~\ref{thm:no_semiuniform-entangling_ham} applies. 

We further show in the S.M., that whenever the irrep is semiuniform (which includes uniform) for any generic group, there is no entangling Hamiltonian with support over $N-1$ sites or fewer, and for all $(N-1)$-local Hamiltonian, $\rho_{\beta, \Lambda}$ is fully separable when $\rho_\beta$ is fully separable. We can also show that semiuniform irreps are rare: if the symmetry is homogeneous and independent of the number of sites, the fraction of semiuniform global irreps goes to zero in the thermodynamic limit as $O(2^d/N)$, where $d$ is the on-site Hilbert space dimension~\cite{supp}. Thus, for most irreps entanglement persists at high temperatures, even if the symmetry group is not finite.

\textit{\textbf{Persistence of negativity---}}
We can also ask whether $\rho_{\beta, \Lambda}$ remains entangled while exhibiting zero negativity. This question was answered negatively in Ref. \cite{kim2025persistenttopologicalnegativityhightemperature} for the classical Ising Hamiltonian $H = \sum_i Z_i Z_{i+1}$ with $\Z_2$ symmetry generated by $\prod_i X_i$. We report in Fig. \ref{fig:cluster_chain_plot} the same conclusion for the cluster chain Hamiltonian $H = \sum_i Z_{i-1} X_i Z_{i+1}$ with the same symmetry.

We generalize these examples by showing that the entanglement condition of Eq. \eqref{eq:entangling_cond} also implies nonzero negativity, and thus the absence of bound entanglement with positive partial transpose (\PPT)~\footnote{The logarithmic negativity upper bounds distillable entanglement; however non-zero negativity does not necessarily imply no bound or undistillable entanglement, as there may exist bound entanglement that is nevertheless non-\textsf{PPT}.}. To establish this, we begin by defining the basis with respect to which the partial transpose is taken.

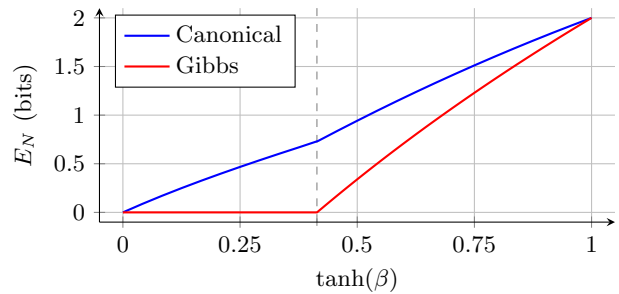
\begin{figure}[t]
    \centering
    \begin{tikzpicture}
    \begin{axis}[
        axis lines = left,
        xlabel={$\tanh(\beta)$}, 
        ylabel=$E_N$ (bits),
        legend pos=north west,
        domain=0:1,
        enlargelimits=0.05,
        y post scale=0.5,
        legend cell align={left},
        grid=both,
        xtick={0, 0.25, 0.5, 0.75, 1},
        ytick={0, 0.5, 1, 1.5, 2},
        minor xtick = {0.414},
        minor grid style = {dashed, gray},
    ]   
        \addplot[thick,blue] 
        {ln(4*(
        3*abs(1+x^2)^2 +
        4*abs(x-x^3) +
        8*abs(x+x^3) +
        2*abs(2-2*x^4) +
        abs(1-6*x^2+x^4)
        )/2^5)
        / ln(2)
        };
        \addlegendentry{Canonical};

        \addplot[thick,red,samples=100]
        {2*ln((
        abs(1+2*x-x^2) + 
        2*abs(1+x^2) +
        abs(1-2*x-x^2)
        )/4)
        / ln(2)
        };
        \addlegendentry{Gibbs};
    \end{axis}
\end{tikzpicture}
    \caption{Entanglement negativity $E_N$ for the Gibbs and canonical ensembles of the cluster chain Hamiltonian $H = \sum_{i=1}^N Z_{i-1} X_i Z_{i+1}$ under $\Z_2$ symmetry generated by $\prod_i X_i$ in the limit $N \to \infty$. Both curves are independent of the entangling interval size $|A| \geq 2$, and the canonical one, of the charge sector. See \cite{supp} for the analytical derivation.}
    \label{fig:cluster_chain_plot}
\end{figure}

For the case of on-site Abelian symmetries $U = \bigotimes_i u_i$, we choose to perform the partial transpose in a region $A$ with respect to a basis that diagonalizes $U_A$, such that $\forall g \in G, U_A(g)^{T_A} = U_A(g)$, and thus we would have $\Pi_{\Lambda}^{T_A} = \Pi_{\Lambda}$. 

Now consider the operator $\tilde{\Pi}_\Lambda\,\rho_{\beta,\Lambda}^{T_A}\,\tilde{\Pi}_\Lambda$. It is traceless, as $\Tr(\tilde{\Pi}_\Lambda\,\rho_{\beta,\Lambda}^{T_A}\,\tilde{\Pi}_\Lambda)\propto \Tr(e^{-\beta H} \Pi_\Lambda \tilde{\Pi}_\Lambda^{T_A})=0$, and if $\rho_{\beta, \Lambda}$ has \PPT, then $\tilde{\Pi}_\Lambda\,\rho_{\beta,\Lambda}^{T_A}\,\tilde{\Pi}_\Lambda$ is also positive semidefinite. Hence, it must be identically zero. Just like before, assuming $\rho_{\beta, \Lambda}$ has \PPT\ in a sequence of $\beta$ converging to zero implies that the derivative of $\tilde{\Pi}_\Lambda\,\rho_{\beta,\Lambda}^{T_A}\,\tilde{\Pi}_\Lambda$ at $\beta = 0$ must vanish, which gives a necessary condition for zero negativity at high temperatures: $(\one - \Pi_{\Lambda}) (H \Pi_{\Lambda})^{T_A} (\one - \Pi_{\Lambda}) = 0$.

Negating this necessary condition gives a sufficient one for \POE\ with non-zero negativity for a bipartition $A|B$:
\begin{align}\label{eq:negativity_cond}
    \left(\one - \Pi_{\Lambda}\right) (H \Pi_{\Lambda})^{T_A} &\left(\one - \Pi_{\Lambda}\right) \neq 0. \nonumber \\ &~\text{[Negativity condition (\textsf{NC})]}
\end{align}

In S.M.~\cite{supp}, we show the equivalence of the entangling and negativity conditions, $\text{EC} \iff \text{NC}$.

\textit{\textbf{Fermions---}} 
We now discuss fermionic systems, where the physically relevant states are always (weakly) symmetric under fermion parity. Consider a local Hilbert space $\mathcal{H}$ and an operator algebra spanned by mutually anticommuting Majorana operators $\{c_j\}$ with $\{c_j, c_k\} = 2\delta_{jk}$, with the fermion parity operator $P \equiv \prod_{i}c_i$. Physical density operators preserve $P$ and their set is referred to as $\mathcal{S}(\mathcal{H})$. Any (symmetrically) separable state under a bipartition $A|B$ must also preserve the local parity operators $P_A, P_B$, where $P_{A(B)} \equiv \prod_{i\in A(B)}c_i$~\cite{Banuls_2007}. Our previous results thus imply that the fermion parity superselection rule generically protects entanglement from sudden death at high temperatures. Here, we establish a stronger result: even a computable entanglement measure, the fermionic negativity~\cite{Shapourian_2017, Shapourian_2019}, generically persists at arbitrarily high temperatures for both the Gibbs and the canonical ensembles.

Any density matrix $\rho \in \mathcal{S}(\mathcal{H}_A \otimes \mathcal{H}_B)$ can be expanded in Majorana monomials $a_{p_1} \cdots a_{p_{k_1}} b_{q_1} \cdots b_{q_{k_2}}$, where $a_{p_j}$ and $b_{q_j}$ act on $A$ and $B$ respectively, and $k_1 + k_2$ is even by fermion parity. The fermionic partial transpose $\mathcal{T}_A$ multiplies such monomial by the phase $i^{k_1}$. The fermionic negativity, defined as $E_{N}^f =\log ||\rho^{\mathcal{T}_A}||_1$ is a more sensitive measure of fermionic entanglement~\cite{Shapourian_2019} than its ``bosonic'' counterpart~\cite{Eisler_2015, Eisert_2018}, which is defined without the phase factor $i^{k_1}$. 

We proceed by contradiction to show that $E_{N}^f$ persists generically for the \emph{Gibbs ensemble} $\rho_\beta \propto e^{-\beta H}$, conjectured in~\cite{Shapourian_2019finite_temperature}. Assume $\rho_{\beta}$ has zero fermionic negativity, i.e. $||\rho_{\beta}^{\mathcal{T}_A}||_1 = 1$, for a sequence of $\beta$ converging to $0$. This, combined with the fact that $\Tr{\rho_{\beta}^{\mathcal{T}_A}} = \Tr{\rho_{\beta}} = 1$, implies that $\rho_{\beta}^{\mathcal{T}_A}$ is Hermitian~\footnote{This is proved in Lemma 2.2. of ~\cite{li_characterizations_2015}: $||A||_1 = \Tr{A}$ if, and only if, $A \geq 0$.}. By taking the derivative of $\rho_{\beta}^{\mathcal{T}_A} - \left(\rho_{\beta}^{\mathcal{T}_A}\right)^{\dagger} = 0$ at $\beta=0$, we have a necessary condition for zero $E_{N}^f$ at high temperature:
\(
H^{\mathcal{T}_A} - (H^{\mathcal{T}_A})^\dagger = 0,
\)
whose negation gives a sufficient condition for the persistence of $E_N^f$,
\begin{equation}
    (H^{\mathcal{T}_A})^\dagger \neq H^{\mathcal{T}_A}.
\end{equation}

This is satisfied by the two-Majorana term \(V = ic_i c_j\), with \(i \in A\) and \(j \in B\), since $(ic_ic_j)^{\mathcal{T}_A}=-c_ic_j$ is not Hermitian. Note that this entangling condition never occurs for the ``bosonic'' negativity, since $H^{T_A}$ is always Hermitian. The same entangling perturbation also leads to the persistence of $E^{f}_{N}$ for the canonical ensemble with fermion parity. A similar idea also resolves a conjecture in~\cite{Shapourian_2019}, that fermionic states that are symmetrically entangled and mix local fermion parity have non-vanishing fermionic negativity, as their partially transposed state is necessarily not Hermitian. We explain both in S.M.~\cite{supp}.

\textit{\textbf{Discussions---}}
Our results imply the persistence of entanglement at high temperatures in two related scenarios: when the thermal state is enforced to be strongly symmetric (canonical ensemble), and when the symmetry is enforced at the level of decomposition of the mixed state into pure states (Gibbs ensemble with superselection). We also discuss the generic nature of this entanglement, showing the existence of local Hamiltonian perturbations which lead to entanglement at arbitrarily high temperatures, as well as a complete classification of charge sectors with \POE\ for homogeneous on-site global Abelian symmetries. Ultimately, genericness follows because the set of symmetrically separable states is contained in a lower dimensional subspace given by linear constraints.

These results can be generalized in various ways. Firstly, the results of the sufficient conditions for \POE\ also extend to Abelian higher form and subsystem symmetries. In fact, strongly symmetric versions of toric codes (with Abelian one form symmetries) already show a persistence of negativity, as can be inferred from~\cite{Hart_2018}. A more careful study of higher form and other symmetries, including antiunitary symmetries, is left for a future study. 

Moreover, our entanglement results rely on exact symmetry conditions and therefore do not capture the more realistic case of only approximately symmetric states, for instance when the symmetry is restored dynamically~\cite{rylands_dynamical_2024}. We expect measures of entanglement asymmetry to be useful in this regard~\cite{ares_entanglement_2023, kusuki_resourcetheoretic_2026}.

Some open questions about the quantification of entanglement at high temperatures remain: our results only imply that a faithful measure of entanglement (and negativity) will remain non-zero at arbitrarily high temperatures for any fixed system size $N$, but it could vanish as $N \to \infty$. Calculations on stabilizer models, however, suggest that the entanglement negativity does not vanish at non-zero $\beta$ even in the thermodynamic limit $N \to \infty$, which will be useful to establish more generally. This is also relevant to the question of local detectability of the entanglement in $\rho_{\beta, \Lambda}$ for temperatures past \SDOE\ of $\rho_{\beta}$. This is also related to the question of local indistinguishability of Gibbs and canonical ensembles at finite temperature, which we have proven for high enough temperatures in the S.M. \cite{supp}.

It is also interesting to explore whether quantum resources beyond entanglement, such as magic and Gaussianity~\cite{ramkumar2025hightemperaturefermionicgibbsstates}, undergo sudden death for the canonical ensemble. Lastly, it would be useful to investigate the complexity of preparation of the canonical ensemble, in view of the recent interest in similar questions about the Gibbs ensemble~\cite{chen2023quantumthermalstatepreparation, chen2023efficientexactnoncommutativequantum,rouze2024efficientthermalizationuniversalquantum, bakshi2025hightemperaturegibbsstatesunentangled}.

\begin{acknowledgments}
    We acknowledge helpful discussions with Amin Moharramipour, Yìlè Yīng, Pablo Sala, Ramanjit Sohal, Farzin Salek, Timothy H. Hsieh, Chong Wang, and Roberto D. Baldijão. L.A.L. and A-R.N. acknowledge support from the Natural Sciences and Engineering Research Council of Canada (NSERC) under Discovery Grants No. RGPIN-2018-04380 (L.A.L. and A-R.N.) and No. RGPIN-2020-04688 (L.A.L.). S.S. thanks the Erwin Schr\"odinger Institute, Vienna, Austria, where part of this work was done. This work was also supported by an Ontario Early Researcher Award. Research at Perimeter Institute is supported in part by the Government of Canada through the Department of Innovation, Science and Industry Canada and by the Province of Ontario through the Ministry of Colleges and Universities.
\end{acknowledgments}

\bibliography{bib}

\newpage

\onecolumngrid

\newpage

\begin{center}
{\Large\textbf{Supplemental Material for}\\[0.25em]
\textbf{``Symmetry enforces entanglement at high temperatures''}\\[1em]}

Amir-Reza Negari$^{1,2,*}$, Leonardo A. Lessa$^{1,2,*}$, and Subhayan Sahu$^{1,\dagger}$\\[0.5em]
$^{1}$Perimeter Institute for Theoretical Physics, Waterloo, Ontario N2L 2Y5, Canada\\
$^{2}$Department of Physics and Astronomy, University of Waterloo, Waterloo, Ontario N2L 3G1, Canada\\[0.75em]
$^{*}$These authors contributed equally to this work.\\
$^{\dagger}$ssahu@perimeterinstitute.ca
\end{center}

\setcounter{page}{1}

The Supplemental Material is organized as follows.
\begin{itemize}
    \item \textbf{Sec.~\ref{appsec:two_qubit}} presents a minimal two-qubit example with $\Z_2$ symmetry that concretely illustrates the persistence of entanglement results of \textbf{Theorem~1}, and highlights how the \EC\ differs from the \SEC\ at small system sizes.

    \item \textbf{Sec.~\ref{appsec:symmetric_decompositions}} gives formal definitions of symmetric separability and proves key structural lemmas, including the equivalence between symmetric separability of the Gibbs ensemble and separability of the canonical ensemble.

    \item \textbf{Sec.~\ref{appsec:local_indistinguishability_b}} establishes that, for finite Abelian groups in the thermodynamic limit, the symmetric entangling condition (\SEC) and the entangling condition (\EC) are equivalent.

    \item \textbf{Sec.~\ref{appsec:irreps}} provides a detailed proof of \textbf{Theorem~2}, introduces the notion of semiuniform irreps, constructs two-body entangling Hamiltonians for non-semiuniform irreps, and shows that semiuniform irreps are rare in the thermodynamic limit.

    \item \textbf{Sec.~\ref{appsec:equivalence_of_entangling_conditions}} proves the equivalence \EC\ $\iff$ \NC, showing that whenever the canonical ensemble is entangled it must also have non-zero entanglement negativity (i.e., no bound entanglement in the \PPT\ sense).

    \item \textbf{Sec.~\ref{appsec:negativity_cluster_chain}} analytically computes the entanglement negativity of both the Gibbs and canonical ensembles for the one-dimensional cluster chain Hamiltonian, confirming sudden death of entanglement for the former and persistence for the latter.

    \item \textbf{Sec.~\ref{appsec:fermions}} extends the persistence of entanglement results to fermionic systems, proving that the fermionic negativity of the canonical ensemble with fixed fermion parity generically persists at arbitrarily high temperatures.

    \item \textbf{Sec.~\ref{appsec:fermionic_negativity_conjecture}} proves the conjecture of Shapourian--Ryu~\cite{Shapourian_2019} that any fermionic state mixing local fermion parity sectors necessarily has non-vanishing fermionic negativity.

    \item \textbf{Sec.~\ref{appsec:local_indistinguishability_a}} proves that the canonical and Gibbs ensembles are locally indistinguishable up to exponentially small errors, both at infinite temperature (for finite Abelian groups) and at sufficiently high but finite temperatures (for $k$-local Hamiltonians).
\end{itemize}

\section{\texorpdfstring{Two qubit example with $\Z_2$ symmetry}{Two qubit example with Z2 symmetry}}\label{appsec:two_qubit}
In this section, we pedagogically examine the case of two qubits with the $\Z_2$ symmetry generated by $X_1 X_2$. Due to its small size, it serves as a simple model that illustrates our main results analytically. However, for the same reason, some large-size results are not valid. In particular, the equivalence of \SEC\ and \EC\ (See Appendix \ref{appsec:local_indistinguishability_b}) as the projection term $\Pi_\lambda$ cannot be removed from the \EC: $[U_A(g), H] \Pi_\Lambda \neq 0$. 

Indeed, consider the 2-qubit XYZ Hamiltonian $H = X_1 X_2 + J(\frac{1+\gamma}{2} Y_1Y_2 + \frac{1-\gamma}{2} Z_1Z_2) = X_1X_2 + J Y_1Y_2 (\gamma \Pi_{+1} + \Pi_{-1})$, with added restriction to the $\Z_2$-symmetric sector $X_1X_2 = +1$. The negativity of its canonical ensemble state $\rho_{\beta,+1} \propto e^{- \beta H} \Pi_{+1}$ can be computed to be $\mathcal{N}(\rho_{\beta,+1}) = \frac{1}{2}|\tanh(\beta J \gamma)|$. Thus, for all temperatures, $\rho_{\beta, +1}$ is separable if $\gamma = 0$. However, if $J \neq 0$, $H$ does not commute with the partial symmetry action and the \SEC\ is satisfied: $[X_1, H] = i J (Z_1 Y_2 - Y_1 Z_2) \neq 0$, implying persistence of symmetric entanglement. If one instead includes the projection $\Pi_{+1}$, then we have $[X_1, H \Pi_{+1}] = 2 iJ \gamma Z_1 Y_2 \Pi_{+1} = 0$, as expected from the separability of $\rho_{\beta, +1}$. One might notice that this counterexample relies on the projections $\Pi_{\pm1}$ directly appearing in the expression of the Hamiltonian, which should not happen for a system of many parties and a local Hamiltonian.

We can visually represent this example in the case of $\gamma = 0$ with the diagram in Fig. \ref{fig:geometry_twoqubits}. It is an accurate representation of a slice in the space of quantum states akin to Fig. \ref{fig:geometry_schematic}. We can see that the Gibbs state (red dashed line, $J=1$) becomes symmetrically entangled as soon as $\beta > 0$, but has sudden death of (non-symmetric) entanglement at low enough $\beta$. At the same time, the canonical ensemble (thick blue line from $\frac{1}{2} \Pi_-$ to $\ket{\Psi_-}$) becomes entangled as soon as $\beta > 0$.

\begin{figure}[ht]
    \centering
    \begin{tikzpicture}[scale=2]
    \fill[orange!10] (0,1) node[above] {$\ketbra{\Psi_-}{\Psi_-}$} 
                  -- (0,-1) node[below] {$\ketbra{\Phi_-}{\Phi_-}$} 
                  -- ($sqrt(3)*(1,0)$) node[right] {$\frac{1}{2} \Pi_+$}
                  -- cycle;

    \draw[fill=green!30!white, draw=green] (0,0) -- ($0.5*sqrt(3)*(1,0)+0.5*(0,1)$) -- ($sqrt(3)*(1,0)$)  node[near start,sloped,above,green!60!black] {\footnotesize\SEP} -- ($0.5*sqrt(3)*(1,0)+0.5*(0,-1)$) -- cycle;
    
    \draw[draw=green!50!black, very thick] (0,0) -- ($sqrt(3)*(1,0)$);
    \draw[green!30!black] (1.13, 0.1) node {\footnotesize \SymSEP};   

    \draw (0,1) node[above] {$\ketbra{\Psi_-}{\Psi_-}$} 
       -- (0,-1) node[below] {$\ketbra{\Phi_-}{\Phi_-}$} 
       -- ($sqrt(3)*(1,0)$) node[right] {$\frac{1}{2} \Pi_+$}
       -- cycle;

    \fill[black] ($0.5*sqrt(3)*(1,0)$) circle [radius=0.6pt] node[below] {$\one/4$};
    \fill[black] (0,0) circle [radius=0.6pt] node[left] {$\frac{1}{2} \Pi_-$};
    \fill[black] (0,1) circle [radius=0.6pt];
    \fill[black] (0,-1) circle [radius=0.6pt];
    \fill[black] ($sqrt(3)*(1,0)$) circle [radius=0.6pt];

    \draw[red, -{Stealth}, thick, dashed] ($0.5*sqrt(3)*(1,0)$) -- ($0.2*sqrt(3)*(1,0) + (0,0.6)$);
    \draw[red, thick, dashed]($0.2*sqrt(3)*(1,0) + (0,0.6)$) -- (0,1);
    \draw [blue, -{Stealth}, very thick] (0,0) -- (0,0.5);
    \draw [blue, very thick] (0,0.5) -- (0,1);
\end{tikzpicture}
    \caption{Subset of quantum states of two qubits formed by convex mixtures of the Bell states $\ket{\Psi_-} \propto \ket{01} -\ket{10}$ and $\ket{\Phi_-} = \ket{00} - \ket{11}$, and the maximally mixed $X_1X_2 = +1$ state $\frac{1}{2} \Pi_+ = \frac{1}{2} \ketbra{00}{00} + \frac{1}{2} \ketbra{11}{11}$. The dashed line traces the path of the Heisenberg antiferromagnetic thermal states $e^{-\beta \mathbf{S}_1 \cdot \mathbf{S}_2}$, or, equivalently, of the Werner states $\lambda \ketbra{\Psi_-}{\Psi_-} + (1-\lambda) \frac{1}{4} \one$. Meanwhile, the thick blue line going from $\frac{1}{2} \Pi_-$ to $\ket{\Psi_-}$ traces the path of the canonical ensemble $e^{- \beta \mathbf{S}_1 \cdot \mathbf{S}_2} \Pi_{-1}$ in the $X_1 X_2 = -1$ sector. The \SEP\ region consists of separable states, with \SymSEP\ line inside it indicating the symmetrically separable ones. Their loci were found via the Peres-Horodecki criterion, which it is necessary and sufficient for two-qubit systems \cite{Peres_1996, Horodecki_1996}. This figure is based on Fig. 16.8(b) from \cite{bengtsson_geometry_2017}.}
    \label{fig:geometry_twoqubits}
\end{figure}
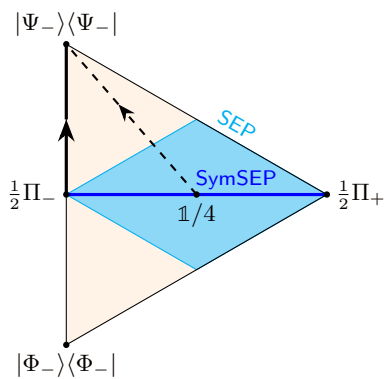

\section{Symmetric entanglement}
\label{appsec:symmetric_decompositions}

We briefly recall the definition of separability and symmetric separability. A state $\rho$ is separable for a bipartition $A|B$ if it can be represented as a convex mixture of tensor product states, $\rho = \sum_{i}p_i \rho_{A,i} \otimes \rho_{B,i}$. A separable state $\rho$ is also symmetrically separable if all $\rho_{A,i} \otimes \rho_{B,i}$ are also (weakly) symmetric.

Following these definitions, and the definitions of weak and strong symmetries, we prove the following lemmas,

\begin{lemma}\label{lemma:equivalence_symsep_sep_SEC}
    If the symmetry action $U : G \to U(\Hilb)$ is on-site, $U = U_A \otimes U_B$, then $\rho$ is symmetrically separable if, and only if, $\rho$ is separable and $\forall g \in G, [U_A(g), \rho] = 0$.
\end{lemma}
\begin{proof}
    $(\Rightarrow)$ follows from the definition of symmetric separability. For the $(\Leftarrow)$ direction, given a weakly symmetric separable state $\rho = \sum_i p_i \rho_{A,i} \otimes \rho_{B,i}$ that commutes with the partial symmetry action $U_A$, then it also commutes with $U_B$ and we have $\rho = \tau_A \otimes \tau_B(\rho)$, where $\tau_A (M) \defeq \frac{1}{|G|} \sum_{g \in G} U_A(g) M U_A(g)^\dagger$ is the twirling operation with respect to $U_A$, and similarly for $\tau_B$. Then, $\rho$ has an explicit symmetrically separable decomposition given by
    \begin{equation}
        \rho = \tau_A \otimes \tau_B(\rho) = \sum_i p_i \tau_A(\rho_{A,i}) \otimes \tau_B(\rho_{B,i}),
    \end{equation}
    where each $\tau_A(\rho_{A,i})$ and $\tau_B(\rho_{B,i})$ state is weakly symmetric due to the twirling operation.
\end{proof}

\begin{lemma}\label{lem:equivalence_of_symm_separability}
    A strongly symmetric state $\rho$ is symmetrically separable if, and only if, it is separable.
\end{lemma}
\begin{proof}
    $(\Rightarrow)$ is true for any state. For the $(\Leftarrow)$ direction, given a separable and strongly symmetric state $\rho = \sum_i p_i \rho_{A,i} \otimes \rho_{B,i}$, each tensor product state $\rho_{A,i} \otimes \rho_{B,i}$ is also strongly symmetric due to the inheritance property of the strong symmetry \cite{lessa_higher-form_2025}. In particular, they are also weakly symmetric.    
\end{proof}
\begin{lemma}\label{lem:equivalence_of_symm_entanglement}
    If the symmetry action $U : G \to U(\Hilb)$ is on-site, $U = U_A \otimes U_B$, then a weakly symmetric state $\rho$ is symmetrically separable if, and only if, each of its strongly symmetric parts $\rho_\lambda \propto \rho \Pi_\Lambda$ is separable.
\end{lemma}
\begin{proof}
$(\Rightarrow)$ Since the irrep projector $\Pi_\Lambda$ decomposes as $\Pi_\Lambda = \sum_{\lambda} \Pi_\lambda^A \Pi^B_{\Lambda \overline{\lambda}}$, then if $\rho = \sum_{i} p_i \rho_{A,i} \otimes \rho_{B,i}$, we have the following separable decomposition of $\rho_\lambda$:
\begin{align}
    \rho_\Lambda \propto \rho \Pi_\Lambda = \sum_{i, \lambda} p_i (\rho_{A,i} \Pi_{\lambda}^A) \otimes (\rho_{B,i} \Pi^B_{\Lambda \overline\lambda}).
\end{align}

$(\Leftarrow)$ follows from $\rho$ being a convex mixture of its strongly symmetric parts $\rho_\lambda$, each of which is separable, and thus symmetrically separable by Lemma \ref{lem:equivalence_of_symm_separability}.  
\end{proof}

The last two lemmas prove our assertion in the main text: $\rho$ is symmetrically entangled (i.e. \textit{not} symmetrically separable) if, and only if, its strongly symmetric component $\rho_{\Lambda}$ is also entangled for some charge $\Lambda$.

\section{\texorpdfstring{Equivalence of \SEC\ and \EC\ for finite Abelian groups}{Equivalence of SEC and EC for finite Abelian groups}}
\label{appsec:local_indistinguishability_b}

Here, we find a sufficient condition on the symmetry multiplicities that allows us to establish equivalence between the entangling condition $(\EC)$ and $(\SEC)$ for finite Abelian groups. First, a technical lemma.

\begin{lemma}\label{lemma:PiLambda_removal}
    Consider a Hamiltonian defined on a Hilbert space with the tripartition $A|B|C$, $\Hilb = \Hilb_A \otimes \Hilb_B \otimes \Hilb_C$, $H = H_{AB} \otimes \one_C + \one_{A} \otimes H_{BC}$, such that $B$ acts as a buffer region between $A$ and $C$ with no global interaction term nor any local interaction term between $A$ and $C$. Furthermore, each term $H_{AB}$ and $H_{BC}$ is assumed to be individually symmetric: $\forall g \in G, [U(g), H_{AB}] = [U(g), H_{BC}] = 0$. If all irreps $\lambda$ that appear in $AB$ (i.e. $m^{AB}_{\lambda} \neq 0$) contribute to the total irrep $\Lambda$ (i.e. $m^C_{\Lambda\overline\lambda} \neq 0$), then we have an equivalence between the $\EC$ and $\SEC$,
    \begin{equation}
        [U_A(g), H]\Pi_\Lambda \neq 0 \iff [U_A(g), H] \neq 0.
    \end{equation}
\end{lemma}
\begin{proof}
    The ($\Rightarrow$) direction is trivial. For the ($\Leftarrow$) direction, we prove their negations are equivalent. We first right-multiply the RHS $[U_A(g), H]\Pi_\Lambda = 0$ by $\Pi^{AB}_\lambda$ where $\lambda$ is any irrep that appears in $\Hilb_{AB}$, i.e. $m_{\lambda}^{AB}\neq0$:
    \begin{align}
        0 & = [U_A(g), H]\Pi_{\Lambda} \Pi^{AB}_{\lambda} = [U_A(g), H] \Pi_{\lambda}^{AB} \Pi_{\Lambda\overline\lambda}^C  = [U_A(g), H_{AB}] \Pi_{\lambda}^{AB} \Pi_{\Lambda\overline\lambda}^C .
    \end{align}
    Then, we trace out the $C$ subsystem:
    \begin{equation}
        0 = [U_A(g), H_{AB}] \Pi_{\lambda}^{AB} \otimes \Tr[\Pi^C_{\Lambda\overline\lambda}] = m^C_{\Lambda\overline{\lambda}}[U_A(g), H_{AB}] \Pi_{\lambda}^{AB}.
    \end{equation}
    By assumption, $m_{\lambda}^{AB} \neq 0$ implies $m^C_{\Lambda\overline\lambda} \neq 0$, which allows us to get rid of $m^C_{\Lambda \overline\lambda}$ from the equation above. Summing over $\lambda$, we finally arrive at
    \begin{align}
        0 & = \sum_{\lambda, m_{\lambda}^{AB}\neq 0} [U_A(g), H_{AB}] \Pi_{\lambda}^{AB} \\
          & = \sum_\lambda \Pi_{\lambda}^{AB} [U_A(g), H_{AB}] \Pi_{\lambda}^{AB} \\
          & = \sum_{\lambda, \lambda'} \Pi_{\lambda}^{AB} [U_A(g), H_{AB}] \Pi_{\lambda'}^{AB} \\
          & = [U_A(g), H_{AB}], \\
    \end{align}
    which, after restoring the $C$ system by tensoring with $\one_C$, we have $[U_A(g), H] = 0$.
\end{proof}

In a multipartite system $\Hilb = \bigotimes_{i \in M} \Hilb_i$ with a Hamiltonian $H = \sum_{X \subseteq M} h_X$ that is a sum of local terms $h_X$, $\supp(h_X) = X$, we usually require them to be individually symmetric, $\forall X \subseteq M, \forall g \in G, [U(g), h_X] = [U_X(g), h_X] = 0$. Here, locality means that for all regions $A$ with size $|A| \leq N/2$, there exists a buffer region $B$ with size $|B| < N/2$ (normally, $|B| \ll N$) for which $H$ has no direct interaction between $A$ and $C = (AB)^C$, i.e. $X \cap A \cap C \neq 0 \Rightarrow h_X =0$. Then the Lemma above can be applied for such $A$ and $B$ regions if its conditions are met. In such case, we know that $\forall A \subseteq M, |A| < N/2 \Rightarrow \exists g \in G, [U_A(g), H] \neq 0$ is equivalent to $\rho_{\beta, \Lambda}$ having nonzero negativity for all bipartitions at high enough temperatures.

For finite Abelian symmetry groups $G$ acting faithfully, then every irrep of $G$ should appear for large enough regions $C$. Moreover, ``large enough'' here should depend only on the group itself, and not on the representation or system size. For such cases, then the above lemma is useful to prove $\POE$ purposes for any finite-range Hamiltonian. Let us make this intuition more precise, starting with the following lemma:
\begin{lemma}[Stable power of generating set]\label{lemma:stable_power_generating_set}
    Let $H$ be a finite group (not necessarily Abelian) and $S \subseteq H$ a generating set of $H = \langle S \rangle$ containing the identity element $e \in S$. With $S^n$ being the set of all products of $n$ elements of $S$, we have
    \begin{equation}
        \{e\} = S^0 \subseteq S^1 \subseteq S^2 \subseteq \cdots \subseteq S^{|H|-1} = S^{|H|} = \cdots = H
    \end{equation}
\end{lemma}
\begin{proof}
    First, note that $m \leq n \implies S^m \subseteq S^n$ since $S^n = S^{m} \cdot S^{n-m} \supseteq S^m \cdot \{ e \} = S^m$. Moreover, since $H$ is finite and $S$ generates it, then, for large enough $n$, $S^n = H$. Hence, it only remains to prove that $\forall n \geq |H|, S^{n-1} = S^{n}$. 
    
    Suppose, by contradiction, that there is an element $h \in S^n \setminus S^{n-1}$. Since $h \in S^n$, there are $s_1, \ldots, s_n \in S$ such that $h = s_1 s_2 \cdots s_n$. Now, recursively define the partial product sequence $(h_k)_{k=0}^n$ by $h_0 = e$ and $h_k = h_{k-1} s_k$. Crucially, such sequence satisfies $h_k \in S^k$ and $h_k^{-1} h \in S^{n-k}$. 
    
    We will prove that the assumption $h \notin S^{n-1}$ implies all elements $h_k$ are distinct from each other. Otherwise, if $h_k = h_l$ for $0 \leq k < l \leq n$, then 
    \begin{equation}
        h = h_l (h_l^{-1} h) = h_k (h_l^{-1}h) \in S^k \cdot S^{n-l} \subseteq S^{n-1},
    \end{equation}
    contradicting the assumption $h \notin S^{n-1}$. However, $H$ has only $|H|$ elements, so the $n+1 \geq |H|+1$ elements of $(h_k)_{k=0}^{n}$ cannot be all distinct. Thus, the original hypothesis of $\exists h \in H, h \in S^n \setminus S^{n-1}$ is incorrect, and we have $S^{n-1} = S^n$.
\end{proof}

The condition $e \in S$ is important. For a counterexample, consider $H = \Z_2 \times \Z_2$ and $S = \{(1,0), (1,1)\}$. Even though $S$ generates $H$, its powers alternate between $S = S^3 = \cdots = \{(1,0),(1,1)\}$ and $S^2 = S^4 = \cdots = \{(0,0), (0,1)\}$, never stabilizing to the full group $H$.

We can apply the lemma above to a subgroup of $\widehat{G}$, the group of 1d irreps of $G$ (also called the Pontryagin dual of $G$), generated by the set of irreps appearing in the on-site representation $u$, here assumed to be site-independent. Let us frame it in the context of spin chains for simplicity, but a similar result will be valid for any finite-range Hamiltonian in any bounded-degree graph (which includes higher dimensional lattices).
\begin{theorem}
    Let $G$ be a finite Abelian group. Consider a spin chain with periodic boundary conditions and homogeneous symmetry action $U = u^{\otimes N} : G \to U(\Hilb)$, with $u$ containing the trivial irrep in its decomposition. Further assume that the terms $h_i$ of the Hamiltonian $H = \sum_{i} h_i$ are supported around site $i \ni \supp(h_i)$, have finite range $d$, $\diam(\supp(h_i)) \leq d$, and are individually symmetric under $U$. For any global irrep $\Lambda$ and any connected region $A$ (an interval) separating $H = H_A \otimes \one_{A^C} + \one_A \otimes H_{A^C} + H_{AA^C}$ with $|A| \leq N - (2d + |G| - 3)  = N - O(1)$, we have
    \begin{equation}
        [U_A, H]\Pi_\Lambda \neq 0 \iff [U_A, H] \neq 0.
    \end{equation}
    
\end{theorem}
\begin{proof}
    Let $S = \{ \lambda \in \widehat{G} \mid \lambda \in u \}$, where ``$\lambda \in u$'' means that $\lambda$ appears in the irrep decomposition of the on-site representation $u$ with nonzero multiplicity. By assumption $e \in S$. Then, it is easy to see that $S^n = \{ \lambda \in \widehat{G} \mid \lambda \in u^{\otimes n} \}$. Lemma \ref{lemma:stable_power_generating_set} guarantees that all irreps that will ever appear in any tensor product of $u$ form the subgroup $\langle S \rangle \leq \widehat{G}$, and that all regions $C$ with size $|C| \geq |\widehat{G}|-1 = |G|-1$ will contain them, i.e. $S^{|C|} = \langle S \rangle$. Here, we used that if $G$ is a finite Abelian group, then $G \simeq \widehat{G}$. That follows from $\widehat{\Z_n} \simeq \Z_n$ and the fact that $G$ is a direct product of cyclic groups. For such regions $C$, we can apply Lemma \ref{lemma:PiLambda_removal} if $B$ is large enough to separate the interval $A$ from $C$. That is because both the irrep $\lambda$ of $AB$ with multiplicity $m_\lambda^{AB} \neq 0$ is an element of $S^{|AB|} \subseteq \langle S \rangle$, and the total irrep $\Lambda$ is an element of $S^{N} = \langle S\rangle$, which means that $\Lambda \overline\lambda \in \langle S \rangle = S^{|C|}$ appears in $C$. For that, we choose $B = \{ i \in A^C \mid \text{dist}(A,i) \leq d-1 \}$, and $C = (AB)^C$. Given that $A$ is connected, these three regions separate Hamiltonian into $H = \one_A \otimes H_{BC} + H_{AB} \otimes \one_C$, thus satisfy all of the assumptions of Lemma \ref{lemma:PiLambda_removal}. Finally, the condition $|C| \geq |G| -1$ translates to $|A| = N - |B| - |C| \leq N - 2(d-1) - (|G| - 1)$.
\end{proof}

As an example, let us consider a nearest-neighbor Hamiltonian $H = \sum_i h_{i,i+1}$ with homogeneous on-site $\Z_2$ symmetry given by $\prod_i u_i$, $u_i^2 = \one_i$ (so $d = 2$ and $|G| = 2$). If the terms $h_{i,i+1}$ are $\Z_2$ symmetric, but not diagonal in the $u_i$ basis, then they will not commute with the partial action of the symmetry, i.e. $[u_i, h_{i,i+1}] \ne 0$. Under these assumptions, the corollary above implies that any canonical ensemble state will not have $\SDOE$ for all bipartitions into two intervals, if the total system size $N$ is larger than $5$. That is because $\lfloor N/2 \rfloor \leq N - (2d + |G| - 3) = N-3$ for $N \geq 5$. 

Note that the finiteness of $G$ is important for the above. For an example of an infinite Abelian symmetry with tighter constraints on the regions $A$ for which Lemma \ref{lemma:PiLambda_removal} can be applied, consider an homogeneous $G=U(1)$ symmetry action with on-site charges $0$ and $+1$, which is the case for $U(1)$ acting on qubits as $\theta \mapsto \diag(1, e^{i \theta})$. Then, the condition of Lemma \ref{lemma:PiLambda_removal} is equivalent to the following condition on the total charge $n$ ($\Lambda(\theta) = \exp(i n \theta)$): $|AB| \leq n \leq N - |AB|$. In particular, the Lemma can be applied to any finite-range Hamiltonian constrained to a total charge $n \approx N - O(1)$, if only entanglement of finite regions $|A| = O(1)$ is relevant. For larger regions, such as $|A| = N/2 - O(1)$, the effect of the global charge constraint can be noticeable, which does not happen for finite Abelian groups (if $N$ is large enough).

\section{Classification of irreps}\label{appsec:irreps}

Here we provide a detailed proof of Theorem~\ref{thm:no_semiuniform-entangling_ham} presented in the main text. The definition of semiuniform irreps and the theorem is quoted below for the ease of readers.

\begin{definition}
    A global irrep $\Lambda$ of a $N$-partite system is \textbf{uniform} if its projector $\Pi_\lambda$ is a uniform tensor product of an on-site irrep $\lambda$: $\Pi_\Lambda = \bigotimes_{i=1}^N \Pi_\lambda^{(i)}$. $\Lambda$ is \textbf{semiuniform} if its projector $\Pi_\Lambda$ is a direct sum of uniform tensor products: $\Pi_\Lambda = \sum_\alpha \bigotimes_{i=1}^N \Pi_{\lambda_\alpha}^{(i)}$. 
\end{definition}

A small-sized example of semiuniform irrep is the trivial irrep $\Lambda = 1$ of a $\Z_3$ action on $N=3$ qubits generated by $\tilde{Z}_1 \tilde{Z}_2 \tilde{Z}_3$, where $\tilde{Z} = \diag(1, e^{2\pi i/3})$. In this case, $\Pi_{\Lambda=1}$ projects onto the subspace generated by the uniform states $\ket{000}$ and $\ket{111}$. An example of an uniform irrep is the maximal or the minimal charge irrep of the aforementioned $U(1)$ symmetry on qubits.

\begin{theorem}
    Consider a multipartite system with homogeneous symmetry action $U$. If $\Lambda$  
    \begin{enumerate}[{(A)}]
        \item is not semiuniform, then for any pair of sites $(i,j)$, there exists a two-body symmetric Hamiltonian $v_{i,j}$ supported on $\{i, j\}$ that entangles any bipartition that separates site $i$ from site $j$. Moreover, given a connected interaction graph, the uniform nearest-neighbor perturbation $\sum_{\langle i, j \rangle} v_{i,j}$ entangling all bipartitions $A|B$, with $A, B \neq \emptyset$.
        \item is semiuniform, then there is no entangling Hamiltonian with support over $N-1$ sites or fewer, and for all $(N-1)$-local Hamiltonian, $\rho_{\beta, \Lambda}$ is fully separable when $\rho_\beta$ is fully separable.
    \end{enumerate}  
\end{theorem}

\begin{proof}

(A) Since $\Lambda$ is not semiuniform, then there is a tensor product projector term $\bigotimes_{i=1}^N \Pi^{(i)}_{\lambda_i}$ in the decomposition of $\Pi_\Lambda$ with at least two distinct irreps over sites $k$ and $l$: $\lambda \neq \lambda'$. Hence, \[\Pi_{\lambda}^{(l)} \Pi_{\lambda'}^{(k)} \Pi_{\Lambda} \neq 0.\] Since $\Pi_\Lambda$ has weak permutation symmetry from the homogeneity of the symmetry action, then we can permute sites $(k, l)$ with $(i, j)$ and $(j, i)$ on the equation above, thus proving that $\Pi_{\lambda}^{(i)} \Pi_{\lambda'}^{(j)} \Pi_{\Lambda}, \Pi_{\lambda'}^{(i)} \Pi_{\lambda}^{(j)} \Pi_{\Lambda} \neq 0$. 

    Thus there exist states $\ket{\lambda}_i$ and $\ket{\lambda'}_{i}$ of $\Hilb_i$ that transform as irreps $\lambda$ and $\lambda'$ under the symmetry, respectively, and similarly with two states $\ket{\lambda}_{j}$ and $\ket{\lambda'}_{j}$ of $\Hilb_j$. Then, let $\sigma_i^+ \defeq {}_i\!\ketbra{\lambda'}{\lambda}_i \in \mathcal{L}(\Hilb_i)$ and $\sigma_j^+ \defeq {}_j\!\ketbra{\lambda'}{\lambda}_j \in \mathcal{L}(\Hilb_j)$. Denoting their Hermitian conjugates by $\sigma_i^- \defeq (\sigma_i^+)^\dagger$ and $\sigma_j^- \defeq (\sigma_j^+)^\dagger$, we choose $v_{i,j} \defeq \sigma_i^+ \sigma_j^- + \sigma_i^- \sigma_j^+$, which is an entangling Hamiltonian because:
    \begin{itemize}
        \item $v_{i,j}$ is symmetric under $U = U_i \otimes U_j \otimes_{k \neq i,j}U_{k}$:
        \begin{align}
            U(g) v_{i,j} U(g)^\dagger = \lambda'(g) \lambda(g) \sigma_i^+ \sigma_j^- \overline{\lambda(g)} \overline{\lambda'(g)} + \text{h.c.} = v_{i,j},
        \end{align}
        
        \item $v_{i,j}$ satisfies the entangling condition \EC\ of Theorem \ref{thm:poe}. To see this, note that the assumptions imply the existence of two states, $\ket{\Psi} = \ket{\lambda}_i \otimes \ket{\lambda'}_j \otimes \ket{\psi}_{\setminus{i,j}}$ and $\ket{\Psi'} = \ket{\lambda'}_i \otimes \ket{\lambda}_j \otimes \ket{\psi}_{\setminus{i,j}}$  that transform as the global irrep $\Lambda$. Moreover, since $\lambda \neq \lambda'$, there exists $g \in G$ such that $\lambda(g) \neq \lambda'(g)$. Now we establish a nonzero matrix element of $[U_{i}(g), H]\Pi_{\Lambda}$:
        \begin{align}
            \braket{\Psi' | [U_i(g), v_{i,j}] \Pi_\Lambda | \Psi} & = (\lambda'(g) - \lambda(g)) \braket{\Psi' | v_{i,j} | \Psi} = \lambda'(g) - \lambda(g) \neq 0.
        \end{align}
    \end{itemize}

We now turn to the uniform nearest-neighbors perturbation $V = \sum_{\{i,j\} \in E} v_{i,j}$, where $E$ is the set of edges of a connected graph $\mathcal{G} = (V, E)$. By contradiction, let us assume it is not entangling for a certain region $A \neq \emptyset,V$. 
    
    In what follows, let $\partial A \defeq \{ \{i,j\} \in E \mid i \in A, j \notin A \}$, and $g \in G$ be a group element for which $\lambda'(g) \neq \lambda(g)$. Since $V$ is not entangling, then we must have,

    \begin{align}
        0 = [U_A(g), V] \Pi_\Lambda & = (\lambda'(g) - \lambda(g)) \Pi_\Lambda  \sum_{\{i,j\} \in \partial A} \sigma^+_i \sigma^-_j + \text{h.c.}. \\
    \end{align}
    Since $A \neq \emptyset, V$ and $\mathcal{G}$ is connected, the sum above contains at least one term corresponding to an edge $\{i^*, j^*\} \in \partial A$. Now, let $\ket{\Psi} = \ket{\lambda}_{i^*} \otimes \ket{\lambda'}_{j^*} \bigotimes_{l \neq i^*, j^*} \ket{\psi_l}_l$ and $\ket{\Psi'} = \ket{\lambda'}_{i^*} \otimes \ket{\lambda}_{j^*} \bigotimes_{l \neq i^*, j^*} \ket{\psi_l}_l$ be tensor product states contained within $\Lambda$, i.e. $\Pi_\Lambda \ket{\Psi} = \ket{\Psi}$. Then,
    \begin{align}
        0 & = \braket{\Psi' | (\lambda'(g) - \lambda(g))\Pi_\Lambda  \sum_{( i, j) \in \partial A} \sigma^+_i \sigma^-_j + \text{h.c.}| \Psi} \\
        & = (\lambda'(g) - \lambda(g))\braket{\Psi' | \sigma^+_{i^*} \sigma^-_{j^*} | \Psi} \\
        & = \lambda'(g) - \lambda(g)
    \end{align}
    a contradiction. Here, only the term $\sigma^+_{i^*} \sigma^-_{j^*}$ survives because the others are either $\sigma^+_{j^*} \sigma^-_{i^*}$ or they contain at least one vertex, say $l \in V$, that is not $i^*$ or $j^*$, and thus will get a diagonal contribution corresponding to the state $\ket{\psi_l}_l$. In both cases, the matrix element is zero by construction.

    Thus, $V = \sum_{\{i,j\} \in E} v_{i,j}$ is an entangling Hamiltonian for all regions $A \neq \emptyset, V$, as desired.

(B) Say $\Lambda$ is semiuniform with projector $\Pi_\Lambda = \sum_\alpha \Pi_\alpha =\sum_\alpha \bigotimes_{i=1}^N \Pi_{\lambda_\alpha}^{(i)}$, where $\Pi_\alpha = \bigotimes_{i=1}^N \Pi^{(i)}_{\lambda_\alpha}$ are orthogonal tensor product projectors. Let $H$ be an arbitrary Hamiltonian that is $G$-symmetric and supported over the $N-1$ first sites, without loss of generality.  Then,
    \begin{align}
        H \Pi_\Lambda & = \Pi_\Lambda H \Pi_\Lambda \\
        & = \sum_{\alpha,\beta} \left[ \left(\bigotimes_{i=1}^{N-1} \Pi^{(i)}_{\lambda_\alpha}\right) H \left(\bigotimes_{i=1}^{N-1} \Pi^{(i)}_{\lambda_\beta}\right) \right] \otimes \Pi^{N}_{\lambda_\alpha} \Pi^N_{\lambda_\beta} \\
        & = \sum_\alpha \left(\bigotimes_{i=1}^{N} \Pi^{(i)}_{\lambda_\alpha}\right) H \left(\bigotimes_{i=1}^{N} \Pi^{(i)}_{\lambda_\alpha}\right). \label{eq:H_local_factorized}
    \end{align}

    From Eq. \eqref{eq:H_local_factorized}, we have that $\alpha \neq \beta \Rightarrow \Pi_\alpha (H \Pi_\Lambda) \Pi_\beta = 0$, which also implies $\Pi_\alpha e^{-\beta' H \Pi_\Lambda} \Pi_\beta = 0$. As such, if $\rho_{\beta'} \propto e^{-\beta' H}$ is separable with decomposition $\rho_{\beta'} = \sum_m p_m \ketbra{\Psi_m}{\Psi_m}$, with $\ket{\Psi_m} = \bigotimes_{i=1}^N \ket{\psi^{(i)}_m}$ product state, then 
    \begin{align}
        \rho_{\beta', \Lambda} & \propto e^{-\beta' H} \Pi_\Lambda \\
        & = \Pi_\Lambda e^{-\beta' H \Pi_\Lambda} \Pi_\Lambda \\
        & = \sum_{\alpha, \beta} \Pi_\alpha e^{-\beta' H \Pi_\Lambda} \Pi_\beta \\
        & = \sum_\alpha \Pi_\alpha e^{-\beta' H} \Pi_\alpha \\
        & \propto \sum_{\alpha, m} p_m (\Pi_\alpha \ket{\Psi_m})(\bra{\Psi_m} \Pi_\alpha), 
    \end{align}
    which gives $\rho_{\beta, \Lambda}$ a fully separable decomposition into the product states $\Pi_\alpha \ket{\Psi_m} = \bigotimes_{i=1}^N (\Pi^{(i)}_{\lambda_\alpha} \ket{\psi^{(i)}_m})$.

\end{proof}

Now, it remains to prove that semiuniform irreps are rare among the set of all global irreps:

\begin{theorem}\label{thm:semiuniform_is_rare}
    Consider a $N$-partite system with an Abelian symmetry group $G$ that 1) does not depend on $N$, 2) acts homogeneously (i.e. $U = u^{\otimes N}$), and 3) acts nontrivially on each site (i.e. $\exists g\in G, u(g) \not\propto \one$, or equivalently there are at least two distinct on-site irreps). Then, the fraction of semiuniform global irreps goes to zero in the thermodynamic limit as $O(2^d/N)$, where $d$ is the on-site Hilbert space dimension. Furthermore, if $G$ is finite, there are no semiuniform global irreps for $N > |G|$.

\end{theorem}
\begin{proof}
    We separate the analysis in two cases: either $\lambda \overline{\mu} \in \widehat{G}$ has finite order for all on-site irreps $\lambda, \mu \in u$, or there is a pair $(\lambda_\infty, \mu_\infty)$ with $\lambda_\infty \overline{\mu_\infty}$ having infinite order.
    \begin{enumerate}
        \item In the first case, define $n < \infty$ be the least common multiple of the orders of all such $\lambda \overline{\mu}$. Note that this is always the case for \textit{finite} symmetry groups $G$, with $n \leq |\widehat{G}| = |G|$. We will prove that there are no semiuniform global irreps for $N > n$, thus trivially satisfying the fraction bound.

        Indeed, if $\Lambda$ is semiuniform with $\Pi_\Lambda \neq 0$, there is at least one $\lambda^*$ one of the irreps that appear as a nonzero uniform term $\bigotimes_{i=1}^N \Pi^{(i)}_{\lambda^*}$ in the sum. Furthermore, since $u$ acts nontrivally, let $\mu \neq \lambda^*$ be another on-site irrep, distinct from $\lambda^*$, with $\Pi_\mu^{(i)} \neq 0$. Then, we can construct a \textit{non-uniform} irrep vector $\vec{\lambda}^{NU}$ given by $\lambda^{NU}_i = \mu$ if $1 \leq i \leq n$ and $\lambda^{NU}_i = \lambda^*$ if $n < i \leq N$. This irrep appears in the decomposition of $\Pi_\Lambda$ because $\mu^{n} (\lambda^*)^{N-n} = (\lambda^*)^{N} = \Lambda$, where we have used that $(\lambda^*)^n = \mu^n =1$, the trivial irrep. This, however, contradicts the assumption of semiuniformity of $\Lambda$.

        \item In the latter case --- of $\lambda_\infty \overline{\mu_\infty}$ having infinite order --- we can lower bound the number of global irreps by choosing some of the on-site irreps to be $\lambda_\infty$, and some to be $\mu_\infty$. As a result, we form a set $\{\lambda_\infty^m \mu_\infty^{N-m}\}_{m=0}^N$ of $N+1$ global irreps that are \textit{pairwise distinct} precisely because $\lambda_\infty \overline{\mu_\infty}$ has infinite order. 
        
        To upper bound the number of \textit{semiuniform} global irreps, first note that there are at most $d$ distinct on-site irreps contained in the decomposition of $u$. Since a semiuniform irrep is a sum over a subset of on-site irreps, there are at most $2^d$ of them. Hence, the fraction of semiuniform global irreps is bounded by $2^{d} / (N+1) < 2^d/N$.
    \end{enumerate}
        
\end{proof}

For some infinite groups, including $U(1)$, the fraction of irreps can be further constrained to be at most $d / (N+1)$ due to the following result:
\begin{proposition}
    If $\widehat{G}$ is torsion-less (e.g. $G=U(1)$), then all semiuniform global irreps $\Lambda$ are uniform.
\end{proposition}
\begin{proof}
    If there are two uniform tensor products, $\bigotimes_i \Pi^{(i)}_{\lambda}$ and $\bigotimes_i \Pi^{(i)}_{\mu}$, contributing to $\Pi_\Lambda$, then this means $(\lambda \overline{\mu})^N = \Lambda \overline{\Lambda} = 1$. Since $\widehat{G}$ is torsion-less, then $\lambda = \mu$. By repeating this argument, we conclude that $\Lambda$ is uniform. 
\end{proof}

\section{\texorpdfstring{Equivalence of $\EC$ and $\NC$}{Equivalence of EC and NC}}\label{appsec:equivalence_of_entangling_conditions}
Here we prove that the conditions \EC\ (eq.~\ref{eq:entangling_cond}) and $\NC$ (eq.~\ref{eq:negativity_cond}) are equivalent. This shows that whenever a symmetric Hamiltonian and irrep satisfies \EC\ to lead to entangled canonical ensembles, the latter must also be negative partial transpose, i.e., their entanglement negativity must be non-zero. We prove the equivalence through the following theorem:
\begin{theorem}\label{thm:equivs_abelian_sep_bipartite}
    The following are equivalent:
    \begin{enumerate}[(i)]
        \item $(\one - \Pi_\Lambda) (H \Pi_\Lambda)^{T_A} (\one - \Pi_\Lambda) = 0$.
        \item $H \Pi_\Lambda = \sum_{\lambda} \Pi_{\lambda}^A (H \Pi_{\Lambda}) \Pi^A_{\lambda} = \sum_{\lambda} \Pi_{\lambda}^A \Pi^B_{\Lambda \overline{\lambda}}H \Pi^A_{\lambda} \Pi^B_{\Lambda \overline{\lambda}}$.
        \item For all $g \in G$, $[U_A(g), H\Pi_\Lambda] = [U_B(g), H\Pi_\Lambda] = 0$. \label{item:separable_bipartite-equiv_cond_commutation}
    \end{enumerate}    
\end{theorem}
\begin{proof}
    \noindent
    \begin{itemize}
    \item $(i) \Rightarrow (ii)$ First, we have
    \begin{equation}
        (H \Pi_\Lambda)^{T_A} = \sum_\lambda \Pi^A_\lambda H^{T_A} \Pi^B_{\Lambda \overline\lambda},
    \end{equation}
    where we have used $(\Pi^A_\lambda)^{T} = \Pi^A_\lambda$, which is true as we choose to do the partial transpose for $A$ and $B$ with respect to a basis that diagonalizes $U_A$ and $U_{B}$, such that $\forall g \in G, U_A(g)^{T_A} = U_A(g)$ and $U_{B}(g)^{T_{B}} = U_{B}(g)$. Then,
    \begin{align}
        0 & = [(\one - \Pi_\Lambda) (H \Pi_\Lambda)^{T_A} (\one - \Pi_\Lambda)]^{T_A} \\
        & \left[\sum_\lambda \Pi^A_\lambda H^{T_A} \Pi^B_{\Lambda \overline\lambda} + \sum_\lambda \Pi^B_{\Lambda \overline\lambda} H^{T_A} \Pi^A_{\lambda} - \sum_\lambda \Pi^B_{\Lambda \overline\lambda} H^{T_A} \Pi^B_{\Lambda \overline\lambda} - \sum_{\lambda} \Pi^A_{\lambda}  H^{T_A} \Pi^A_{\lambda} \right]^{T_A} \\
        & = H \Pi_\Lambda + \Pi_\Lambda H - \sum_\lambda \Pi^B_{\Lambda \overline\lambda} H \Pi^B_{\Lambda \overline\lambda} - \sum_{\lambda} \Pi^A_{\lambda}  H \Pi^A_{\lambda}. 
    \end{align}
    By multiplying the equation above with $\Pi_\Lambda$ and using that $[H, \Pi_\Lambda] = 0$, we have
    \begin{equation}
        H \Pi_\Lambda = \frac{1}{2}\left(\sum_\lambda \Pi_\Lambda\Pi^B_{\Lambda \overline\lambda} H \Pi^B_{\Lambda \overline\lambda} \Pi_\Lambda + \sum_{\lambda} \Pi^A_{\lambda}  H \Pi_\Lambda \Pi^A_{\lambda}\right) = \sum_{\lambda} \Pi^A_{\lambda}  (H \Pi_\Lambda) \Pi^A_{\lambda}.
    \end{equation}

    \item $(ii) \Rightarrow (iii)$
    From assertion (ii), we know that $H \Pi_\Lambda$ is block diagonal with respect to the irreps of $U_A$. Thus, by the converse of Schur's lemma, $H \Pi_{\Lambda}$ commutes with all symmetry operators $U_A(g)$, for all $g \in G$. More explicitly, from $U_A(g) \Pi_\lambda^A = \Pi_\lambda^A U_A(g) = \lambda(g) \Pi^A_\lambda$, we have
    \begin{equation}
        [U_A(g), H \Pi_{\Lambda}] = \sum_{\lambda} \lambda(g) \Pi_{\lambda}^A (H \Pi_{\Lambda}) \Pi^A_{\lambda} - \sum_\lambda \Pi_{\lambda}^A (H \Pi_{\Lambda}) \Pi^A_{\lambda} \lambda(g) = 0.
    \end{equation}

    \item $(iii) \Rightarrow (i)$ Condition $(iii)$ also implies that for all irreps $\lambda$ and $\mu$, $[\Pi^A_\lambda, H \Pi_\Lambda] = [\Pi^B_{\mu}, H \Pi_\Lambda] =0$. Hence,
    \begin{equation}
        H \Pi_\Lambda = H \sum_\lambda \Pi^A_\lambda \Pi^B_{\Lambda \overline\lambda} = \sum_\lambda \Pi^A_\lambda H \Pi^B_{\Lambda \overline\lambda}. 
    \end{equation}
    By taking the partial transpose of the equation above, we get
    \begin{equation}
        (H \Pi_\Lambda)^{T_A} = \sum_\lambda H^{T_A} \Pi^A_{\lambda} \Pi^B_{\Lambda \overline\lambda} = H^{T_A} \Pi_\Lambda, 
    \end{equation}
    which, upon right multiplication with $(\one - \Pi_\Lambda)$, results in zero.

    \end{itemize}
\end{proof}

Note that the negation of the statements $(i)$ and $(iii)$ are the \NC\ and \EC, respectively, thereby proving their equivalence.

\section{Entanglement negativity of the thermal cluster chain state}\label{appsec:negativity_cluster_chain}
We analytically compute the entanglement negativity of the Gibbs and canonical ensembles of the one-dimensional cluster chain Hamiltonian. This model was chosen for both its analytical simplicity, and for the Gibbs ensemble showing non-trivial \SDOE, that is, it is entangled at low temperatures, and becomes separable at high temperatures. This last property makes the cluster chain a more realistic example of the generic case, and is not shared with the classical Ising model studied in \cite{kim2025persistenttopologicalnegativityhightemperature}. 

We confirm that the Gibbs state of the cluster chain exhibits sudden death of entanglement negativity, while the negativity of the canonical ensemble persists for arbitrarily high temperatures. Moreover, we find that there is no bound entanglement for the Gibbs ensemble.

\subsection{Gibbs ensemble}

Consider a connected region $A$ with size $|A| \geq 2$ and its complement $A^C = B$ with size $|B| \geq 2$, so that the total system size $N$ is even. Region $A$ has two endpoints: $a_L$ and $a_R$, and they are adjacent to the endpoints $b_L$ and $b_R$ of B, respectively. Define $A_{\text{int}} \defeq A \setminus \{a_L, a_R\}$ to be the interior of region $A$, with endpoints $a_L$ and $a_R$ removed, and similarly with $B_{\text{int}}$. 

To simplify the entanglement calculation, we will conjugate the cluster chain Hamiltonian $H_{\text{cc}} = - \sum_i Z_{i-1} X_i Z_{i+1}$ by $CZ_{\text{int}}$, defined as
\begin{equation}
    CZ_{\text{int}} \defeq CZ_{a_L, b_L} CZ_{a_R, b_R} \prod_{\langle i, j\rangle} CZ_{i,j} =  \prod_{\substack{\langle i, j\rangle \neq \\ \langle a_L, b_L\rangle, \\ \langle a_R, b_R\rangle}} CZ_{i,j}.
\end{equation}
Since $CZ_{\text{int}}$ is a local unitary to the $A|B$ bipartition, it preserves the entanglement value of the thermal state. This results in the trivial paramagnetic Hamiltonian everywhere except on the two boundary links:
\begin{align}
    H'_{cc} \defeq CZ_{\text{int}} H_{cc} CZ_{\text{int}} & = H'_L + H'_R + H'_{pm},
\end{align}
where
\begin{align}
    H'_L &= -X_{a_L} Z_{b_L} -Z_{a_L} X_{b_L}, \\
    H'_R & =-X_{a_R} Z_{b_R} -Z_{a_R} X_{b_R}, \\
    H'_{pm} & =- \sum_{i \in A_{\text{int}}} X_i - \sum_{i \in B_{\text{int}}} X_i.
\end{align}
The Gibbs state $\rho_\beta' = e^{- \beta H'_{cc}} / \mathcal{Z_\beta}$ is, then,
\begin{align}
    \rho_\beta' & = (e^{- \beta H'_{pm}} / \mathcal{Z}^{pm}_\beta) \otimes (\one + \lambda X_{a_L} Z_{b_L} + \lambda Z_{a_L} X_{b_L} + \lambda^2 Y_{a_L} Y_{b_L}) \otimes (\one + \lambda X_{a_R} Z_{b_R} + \lambda Z_{a_R} X_{b_R} + \lambda^2 Y_{a_R} Y_{b_R}) / 2^4, \\
    \mathcal{Z}^{pm}_\beta & = (2 \cosh(\beta))^{N-4},
\end{align}
where $\lambda \defeq \tanh(\beta)$ and $N$ is the system size. Its partial transpose is, then
\begin{equation}
    (\rho_\beta')^{T_A} =  (e^{- \beta H'_{pm}} / \mathcal{Z}^{pm}_\beta) \otimes (\one + \lambda X_{a_L} Z_{b_L} + \lambda Z_{a_L} X_{b_L} - \lambda^2 Y_{a_L} Y_{b_L}) \otimes (\one + \lambda X_{a_R} Z_{b_R} + \lambda Z_{a_R} X_{b_R} - \lambda^2 Y_{a_R} Y_{b_R}) / 2^4,
\end{equation}
and the logarithmic negativity $E_N = \log \norm{\cdot}_1$ is
\begin{align}
    E_N(\rho_\beta) = E_N(\rho_\beta') & = \log \left(\norm{\one + \lambda X_{a_L} Z_{b_L} + \lambda Z_{a_L} X_{b_L} - \lambda^2 Y_{a_L} Y_{b_L}}_1^2 /2^4\right) \\
    & = 2 \log \Bigl(\sum_{a,b=\pm1} |1 + \lambda (a+b) - \lambda^2 ab|/4\Bigr).
\end{align}
As shown in the red curve of Fig. \ref{fig:cluster_chain_plot}, the entanglement negativity is zero for small $\lambda \leq \lambda_c$. Indeed, it can be easily deduced from the expression above that $ \lambda_c = \sqrt{2}-1$. Moreover, we can further argue that the state is exactly separable for $\lambda \leq \lambda_c$ by showing an explicit separable decomposition for it.

Consider the density matrix \(\rho_{AB}\) defined by
\begin{equation}
    \rho_{\beta} \propto e^{-\frac{\beta}{2} H_A} \, e^{-\frac{\beta}{2} H_B} \, \rho_{AB} \, e^{-\frac{\beta}{2} H_B} \, e^{-\frac{\beta}{2} H_A},
\end{equation}
where $H_{A}$ contains the terms of $H$ entirely supported on region $A$, and similarly for $H_{AB}$. If $\rho_{AB}$ can be decomposed into an ensemble of states that are separable between $A$ and $B$, then a separable decomposition for the full density matrix \(\rho_\beta\) can also be constructed.

In the case of the cluster chain Hamiltonian, and expressing $\rho_{AB}$ only at the four closest sites near the rightmost boundary between $A$ and $B$ for simplicity, we have
\begin{equation}
    \rho_{AB} = \frac{1}{2^4} \left(\one + \lambda Z_{a_R-1}X_{a_R}Z_{b_R} \right) \left(\one + \lambda Z_{a_R}X_{b_R}Z_{b_R+1} \right).
\end{equation}
It is easy to check that $\rho_{AB} = p_1 \rho_1 + p_2 \rho_2 + p_3 \rho_3$, with
\begin{align}
    \rho_1 &= \frac{1}{2^4} \left(\one + (\lambda^2 + 2\lambda) Z_{a_R-1}X_{a_R}Z_{b_R} \right), \\
    \rho_2 &= \frac{1}{2^4} \left(\one + (\lambda^2 + 2\lambda) Z_{a_R}X_{b_R}Z_{b_R+1} \right), \\
    \rho_3 &= \frac{1}{2^4} \left(\one + (\lambda^2 + 2\lambda) Z_{a_R-1}X_{a_R}Z_{b_R} \cdot Z_{a_R}X_{b_R}Z_{b_R+1} \right),
\end{align}
and probabilities $p_1 = p_2 = \lambda/(\lambda^2 + 2 \lambda)$ and $p_3 = \lambda^2/(\lambda^2 + 2 \lambda)$.

Each individual $\rho_i$ is in fact a fully separable matrix since it has the form $\one+\alpha \hat{P}$, where $\hat{P}$ is a Pauli operator. However, $\rho_i$ is positive-semidefinite if and only if: 
\begin{equation}
    -1 \leq \lambda^2+2\lambda \leq 1.
\end{equation}
Therefore the state is separable between left and right if and only if $\lambda \leq \lambda_c = \sqrt{2}-1$.

\subsection{Canonical ensemble}

Here, we will repeat the calculation for the canonical ensemble of the 1d cluster chain. Being a nontrivial SPT, the cluster chain exhibit edge modes that become entangled when global strong symmetry is imposed. Since we are interested in the bulk entanglement, we will disregard this contribution by considering periodic boundary conditions, forming a 1d ring. Hence, a nontrivial connected region $A$ will have two boundary points. Furthermore, contrary to the 1d Ising model, the stabilizers of the cluster chain and the symmetry operators are not independent. As we will see, this correlates the degrees of freedom inside the bulk with the ones at the boundary, and also between the boundary points.

Even though the cluster chain has $\Z_2 \times \Z_2$ symmetry generated by $\prod_{i} X_{2i}$ and $\prod_{i} X_{2i-1}$, we will only impose the strong symmetry of the diagonal component $P \defeq \prod_i X_i = \Lambda \in \{\pm1\}$ with projector $\Pi_\Lambda = (\one + \Lambda P)/2$, for simplicity. After conjugation by $CZ_{\text{int}}$, the symmetry becomes $P' = P_{\text{int}} P_{LR}$, where $P_{\text{int}} = \prod_{i \in A_{\text{int}}} X_i \prod_{i \in B_{\text{int}}} X_i$ and $P_{LR} = Y_{a_L}Y_{b_L}Y_{a_R}Y_{b_R}$. The strongly symmetric thermal state $\rho_{\beta,\Lambda}' = e^{-\beta H'} \Pi_\Lambda'/\mathcal{Z}_{\beta, \Lambda}$ is 
\begin{align}
    \rho_{\beta, \Lambda}' & = \Pi_{\Lambda}' \cdot (e^{- \beta H_{pm}} / \mathcal{Z}^{pm'}_{\beta, \Lambda}) \otimes (\one + \lambda X_{a_L} Z_{b_L} + \lambda Z_{a_L} X_{b_L} + \lambda^2 Y_{a_L} Y_{b_L}) \otimes (L \leftrightarrow R) /2^4, \text{ where} \\
    \mathcal{Z}^{pm'}_{\beta, \Lambda} & = (2 \cosh(\beta))^{N-4} (1 + \Lambda \lambda^N)/2.
\end{align}

Now, $\rho_{\beta, \Lambda}'$ does not factorize into a tensor product over the interior of each region and their boundaries. To deal of this, we decompose $\Pi_{\Lambda}'= \Pi^{\text{int}}_{+1} \Pi^{LR}_{\Lambda} + \Pi^{\text{int}}_{-1} \Pi^{LR}_{-\Lambda}$. Then,
\begin{align}
    \rho_{\beta, \Lambda}' & = \sum_{\mu =\pm1} r_{N,\mu} \rho_{\beta, \mu}^{pm} \otimes \Pi^{LR}_{\mu \Lambda} (\one + \lambda X_{a_L} Z_{b_L} + \lambda Z_{a_L} X_{b_L} + \lambda^2 Y_{a_L} Y_{b_L}) \otimes (L \leftrightarrow R) /2^4,  \text{ where}\\
    \rho^{pm}_{\beta, \mu} & = \Pi^{\text{int}}_{\mu} e^{- \beta H_{pm}} / \mathcal{Z}^{pm}_{\beta, \mu}, \\
    \mathcal{Z}^{pm}_{\beta, \mu} & = (2 \cosh(\beta))^{N-4} (1 + \mu \lambda^{N-4})/2, \\
    r_{N,\mu} & = \mathcal{Z}^{pm}_{\beta, \mu} / \mathcal{Z}^{pm'}_{\beta, \Lambda} = (1 + \mu \lambda^{N-4})/(1+\Lambda \lambda^N).
\end{align}
The ratio $r_{N, \mu}$ converges to $1$ in the thermodynamic limit $N \to \infty$ if $\lambda < 1$ ($\beta < \infty$). Hence, we will set it to 1 in what follows. 

The partial transpose is
\begin{align}
    (\rho_{\beta, \Lambda}')^{T_A} = \sum_{\mu = \pm 1} \rho_{\beta, \mu}^{pm} \otimes [ & (\one + \lambda X_{a_L} Z_{b_L} + \lambda Z_{a_L} X_{b_L} - \lambda^2 Y_{a_L} Y_{b_L}) \otimes (L \leftrightarrow R) + \\ & \mu \Lambda (\one - \lambda X_{a_L} Z_{b_L} - \lambda Z_{a_L} X_{b_L} - \lambda^2 Y_{a_L} Y_{b_L}) \otimes (L \leftrightarrow R) P_{LR} ] / 2^5,
\end{align}
so the logarithmic negativity is
\begin{align}\label{eq:ent_neg_cc_ss}
    E_{N}(\rho_{\beta, \Lambda}) &= \log \Bigl( \sum_{\mu = \pm1} \sum_{\substack{a_1, b_1, \\ a_2, b_2 = \pm 1}} \Bigl| \prod_{i=1,2} (1 + \lambda (a_i + b_i) - \lambda^2 a_i b_i) + \mu a_1 b_1 a_2 b_2 \prod_{i=1,2} (1 - \lambda (a_i + b_i) - \lambda^2 a_i b_i) \Bigr| /2^5\Bigr).
\end{align}
This result is plotted on the blue curve of Fig. \ref{fig:cluster_chain_plot}, which shows the persistency of entanglement for the canonical ensemble for arbitrarily small $\beta$. Also note that the final expression is independent of the global charge $\Lambda$. For finite size, however, the entanglement will depend on $\Lambda$.

The analytical results for the entanglement negativity were validated with small size ($4 \leq N \leq 6$) numerical calculations.

\section{Persistence of fermionic negativity in the canonical ensemble}\label{appsec:fermions}
Consider the \emph{canonical ensemble} for fermions with strong fermion parity, \(\rho_{\beta,\Lambda} \propto e^{-\beta H} \Pi_\Lambda\), where \(\Pi_\Lambda = \frac{1}{2}(\one + \Lambda P)\) projects onto a fixed fermion parity sector with $\Lambda = \pm 1$. Assume $\rho_{\beta, \Lambda}$ has zero fermionic negativity, i.e. $||\rho_{\beta,\Lambda}^{\mathcal{T}_A}||_1 = 1$, for a sequence of $\beta$ converging to $0$. This, combined with the fact that $\Tr{\rho_{\beta,\Lambda}^{\mathcal{T}_A}} = \Tr{\rho_{\beta,\Lambda}} = 1$, imply that $\rho_{\beta,\Lambda}$ is Hermitian and positive semi-definite~\cite{li_characterizations_2015}.

Under the fermionic partial transpose, \(P^{\mathcal{T}_A} = \eta_A P\), with \(\eta_A = (-1)^{|A| \bmod 2}\), implying \(\Pi_\Lambda \Pi_{-\eta_A \Lambda}^{\mathcal{T}_A} = \Pi_\Lambda \Pi_{-\Lambda} = 0\). Thus, the matrix $\Pi_{-\eta_A \Lambda}\rho_{\beta,\Lambda}^{\mathcal{T}_A} \Pi_{-\eta_A \Lambda}$, which is positive semidefinite, is also traceless
\(
\Tr\!\left(\rho_{\beta,\Lambda}^{\mathcal{T}_A} \Pi_{-\eta_A \Lambda}\right) \propto \Tr\!\left(e^{-\beta H} \Pi_\Lambda \Pi_{-\eta_A \Lambda}^{\mathcal{T}_A}\right) = 0
\), and thus identically zero. Since the matrix is holomorphic function of $\beta$, its derivative at $\beta = 0$ must also vanish. Negating the chain of arguments, we find that $E_{N}^f$ persists if the derivative at \(\beta = 0\) is nonvanishing: 
\begin{equation}
\Pi_{-\eta_A \Lambda} (H \Pi_\Lambda)^{\mathcal{T}_A} \Pi_{-\eta_A \Lambda} \neq 0.
\end{equation}
This is the fermionic analogue of the $\NC$, which is also generically satisfied. The two-Majorana term \(V = ic_i c_j\), with \(c_i \in A\) and \(c_j \in B\), obeys \((V \Pi_\Lambda)^{\mathcal{T}_A} = i V \Pi_{-\eta_A \Lambda}\), yielding
\(
\Pi_{-\eta_A \Lambda}(V \Pi_\Lambda)^{\mathcal{T}_A} \Pi_{-\eta_A \Lambda} = i V \Pi_{-\eta_A \Lambda} \neq 0,
\)
thereby providing an explicit entangling perturbation.

\section{\texorpdfstring{Proof of Conjecture made in Shapourian–Ryu~\cite{Shapourian_2019}}{Proof of Conjecture made in Shapourian-Ryu}}\label{appsec:fermionic_negativity_conjecture}

Recall that the fermionic (logarithmic) negativity is
\[
E^{f}_{N}(\rho) \;=\; \log \,\mathrm{Tr}\,\bigl| \rho^{\mathcal{T}_A} \bigr|.
\]

Ref.~\cite{Shapourian_2019} introduces two types of fermionic states (see their Remark~2):
\begin{itemize}
\item \textbf{Type I:} Fermion-number parity of subsystems is even, $[P_A,\,\rho]=0$.
\item \textbf{Type II:} Fermion-number parity of subsystems is mixed, $[P_A,\,\rho]\neq 0$.
\end{itemize}
Type~I states may be (symmetrically) separable or entangled, and, in both cases, the fermionic and bosonic negativities coincide. In contrast, Type~II density matrices are always \emph{entangled}. Ref.~\cite{Shapourian_2019} further conjectured that no type~II state that is symmetrically entangled has vanishing fermionic negativity (their Conjecture~1). Here, we confirm this expectation. 

We first introduce the following lemma, proved in \cite{li_characterizations_2015}.

\begin{lemma}
\label{trace_equal_to_one}
For $A$ a square matrix, if $\|A\|_1 \equiv \mathrm{Tr}\bigl|A\bigr| = \mathrm{Tr}\,A$, then $A \geq 0$. In particular, $A$ is Hermitian.
\end{lemma}

\begin{theorem}
For any bipartite fermionic state $\rho$ of type~II, if $\rho$ is symmetrically entangled, then $E^f_N(\rho) > 0$.

\end{theorem}

\begin{proof}
We prove the converse by assuming $E^{f}_{N}(\rho) = 0$. By definition, this implies $\mathrm{Tr}\bigl|\rho^{\mathcal{T}_A}\bigr| = 1$. Since the fermionic partial transpose preserves the trace, $\mathrm{Tr}\,\rho^{\mathcal{T}_A} = \mathrm{Tr}\,\rho = 1$. By Lemma~\ref{trace_equal_to_one}, $\rho^{\mathcal{T}_A}$ is therefore Hermitian (in fact, positive semi-definite).

We now show that Hermiticity of $\rho^{\mathcal{T}_A}$ contradicts the assumption that $\rho$ is type~II. Expand $\rho$ in Majorana monomials $a_{p_1}\cdots a_{p_{k_1}}\, b_{q_1}\cdots b_{q_{k_2}}$, where $a_{(\cdot)}$ act on $A$ and $b_{(\cdot)}$ on $B$. Being type~II means that in this expansion there exists at least one term with an \emph{odd} number of Majorana operators on $A$. However, such a term acquires a factor of $\pm i$ under the fermionic partial transpose on $A$ (see Eq.~(23) of Ref.~\cite{Shapourian_2019}), which generates purely imaginary contributions in $\rho^{\mathcal{T}_A}$ that preclude its Hermiticity.

\end{proof}

\section{Local indistinguishability of the canonical and Gibbs ensembles}
\label{appsec:local_indistinguishability_a}

In this section, we explore whether the canonical ensemble $\rho_{\beta, \Lambda}$ is locally indistinguishable from the Gibbs ensemble $\rho_\beta$, up to an exponentially small error. We establish this for $\beta = 0$ for finite Abelian groups --- which we use to develop intuition about the equivalence of $\EC$ and $\SEC$ in the main text --- and for small enough $\beta < \beta_*$ if $H$ is a $k$-local Hamiltonian. 

We start by collecting some known facts about the symmetries and their representations. From the symmetry action of $G$ on the Hilbert space $\Hilb$, it gets decomposed into $\Hilb = \bigoplus_\lambda V_{\lambda}^{\oplus m_\lambda}$, where the direct sum is over all irreps $\lambda$ of $G$, with $m_\lambda$ being the multiplicities. The projectors onto $V_{\lambda}^{\oplus m_{\lambda}} \subseteq \Hilb$ can be expressed as
\begin{align}\label{eq:pi_decomp}
    \Pi_{\lambda} & = \frac{1}{|G|} \sum_{g \in G} \overline{\lambda(g)} U(g),
\end{align}
where $\lambda(g)$ is the character. They form an orthogonal resolution of identity. Namely,
\begin{align}
    \sum_\lambda \Pi_\lambda & = \one, \quad \Pi_{\lambda} \Pi_{\lambda'} = \delta_{\lambda, \lambda'}\Pi_\lambda.
\end{align}

For a bipartition $\Hilb = \Hilb_A \otimes \Hilb_B$, each Hilbert space $\Hilb_A$ and $\Hilb_B$ also decomposes into a direct sum of irrep spaces of $U_A$ and $U_B$:
\begin{align}
    \Hilb_A  \simeq \bigoplus_{\lambda} V_{\lambda}^{\oplus m^A_{\lambda}}; ~
    \Hilb_B  \simeq \bigoplus_{\mu} V_{\lambda}^{\oplus m^A_{\mu}}
\end{align}
where both direct sums are over all 1d irreps $\lambda$ and $\mu$ of $G$, with multiplicities $m_{\lambda}^A$ and $m_{\mu}^B$, respectively. Note that the multiplicities can be zero. Similar constructions as $\ref{eq:pi_decomp}$ can be done for projectors onto $V_{\lambda}^{\oplus m^A_{\lambda}} \subset \Hilb_A$ and $V_{\lambda}^{\oplus m^B_{\lambda}} \subset \Hilb_B$.

Two irreps, $\lambda$ and $\mu$, fuse together in the tensor product to an irrep $\lambda \otimes \mu \simeq \lambda \cdot \mu$ of $U = U_A \otimes U_B$, such that
\begin{align}
    \Hilb & \simeq \left(\bigoplus_{\lambda} V_{\lambda}^{\oplus m^A_{\lambda}}\right) \otimes \left(\bigoplus_{\mu} V_{\lambda}^{\oplus m^A_{\mu}}\right) \simeq \bigoplus_{\lambda, \mu} V_{\lambda\mu}^{\oplus(m^A_{\lambda} + m^B_\mu)} = \bigoplus_{\Lambda} V_\Lambda^{\oplus m_\Lambda},
\end{align}
where 
\begin{equation}
    M_\Lambda = \sum_{\lambda, \mu} \delta_{\Lambda,\lambda \mu} (m_\lambda^A + m_{\mu}^B) = \sum_{\lambda} (m^A_\lambda + m^B_{\Lambda \overline\lambda}).
\end{equation}

Similarly, the projector onto the irrep $\Lambda$ of $\Hilb$ decomposes as, 
\begin{align}
    \Pi_{\Lambda} & = \sum_{\lambda} \Pi^A_{\lambda} \otimes \Pi^B_{\Lambda \overline\lambda} \\
    & = \frac{1}{|G|^2} \sum_{g, h \in G} \left(\sum_{\lambda}\overline{\lambda(g)} \lambda(h) \right) \overline{\Lambda(h)}  U_A(g) \otimes U_B(h) \\
    & = \frac{1}{|G|} \sum_{g \in G} \overline{\Lambda(g)} U(g),
\end{align}
where, in the third line, we used Schur's orthogonality relations for Abelian groups: $\sum_{\lambda}\overline{\lambda(g)} \lambda(h) = |G| \delta_{g,h}$. Similar identities can be derived for multipartitions.

With these tools we can prove the following lemma:

\begin{lemma}\label{lemma:trace_irrep_projector}
    Given a finite Abelian group $G$, and an homogeneous symmetry action $U = u^{\otimes N}$, with $u$ 
    
    acting on a $d$-dimensional on-site Hilbert space, then, as $N \to \infty$,
    \begin{equation}
        \Tr[\Pi_\Lambda] = \frac{d^N}{|G / K|} (1 + O(e^{-N/\xi})),
    \end{equation}
    where $K = \{ g \in G \mid u(g) \propto \one \}$, and $\xi = - 1/ \log(\max_{g \notin K} \left|\Tr[u(g)]\right|/d)$.
\end{lemma}
\begin{proof}
    Using that $\Pi_\Lambda = \sum_{g \in G} \overline{\Lambda(g)} U(g)$, we have
    \begin{equation}
        \Tr[\Pi_\Lambda] = \frac{1}{|G|} \sum_{g \in G} \overline{\Lambda(g)} \Tr[u(g)]^N.
    \end{equation}
    From the above, we can see that the behavior of $\Tr[\Pi_\Lambda]$ as $N \to \infty$ is dictated by the largest values of $\left|\Tr[u(g)]\right|$. Since $u(g)$ is unitary, then $\left|\Tr[u(g)]\right| \leq d$ for all $g \in G$, and $\left|\Tr[u(g)]\right| = d$ if, and only if, $u(g) \propto \one$. The set of elements satisfying this relation form the normal subgroup $K = \{ g \in G \mid u(g) \propto \one \}$, and the value of $\Lambda$ there is fixed by the equation $\Lambda(h) \Pi_\Lambda = U(h) \Pi_\Lambda = (\Tr[u(h)]/d)^N \Pi_\Lambda$, where $h \in K$. In the last step, we have used that $u(h) = \Tr[u(h)]\mathbb{I}_d/d$, from the definition of the group $K$.
    
    Finally, let $r = \max_{g \notin K} \left|\Tr[u(g)]\right| < d$ be the second largest absolute value of $\Tr[u(g)]$ in the following:
    \begin{align}
        \Tr[\Pi_\Lambda] & = \frac{d^N}{|G|} \sum_{h \in K} \overline{\Lambda(h)} \left( \frac{\left|\Tr[u(h)]\right|}{d} \right)^N + O(r^N) \\
        & = \frac{d^N}{|G/K|} (1 + O[(r/d)^N]).
    \end{align}
    
\end{proof}

If $\Tr[u(g)] = 0$ for all $g \in G$ satisfying $u(g) \not\propto \one$, then $\xi = +\infty$, and the exponentially small error above is actually zero. This is the case of $\Z_d$ acting on a qudit with $u(1) = Z$ or any other qudit Pauli. 

\begin{theorem}
    The maximally mixed symmetric state $\rho_{\beta=0, \Lambda}$ of a finite abelian group $G$ is locally indistinguishable from the identity state, up to an exponentially small error in the trace norm:
    \begin{equation}
        \norm{\Tr_{A^C}[\rho_{\beta=0,\Lambda}] - \one_A/d^{|A|}}_1 = O(e^{-(N-|A|)/\xi}),
    \end{equation}
    with $\xi$ given in Lemma \ref{lemma:trace_irrep_projector}.
\end{theorem}
\begin{proof}
    Let $B = A^C$. First note that
    \begin{align}
        \Tr_{B}[\Pi_\Lambda] & = \sum_{\lambda} \Pi^A_{\lambda} \Tr_B[\Pi^B_{\Lambda \overline{\lambda}}] \\
        & = \one_A \frac{d^{|B|}}{|G / K|} (1 + O(e^{-|B|/\xi})),
    \end{align}
    where in the second line we used Lemma \ref{lemma:trace_irrep_projector} and $\sum_\lambda \Pi^A_\lambda = \one_A$.

    Then, we have
    \begin{align}
        \Tr_{B}[\rho_{\beta = 0, \Lambda}] & = \frac{\Tr_B[\Pi_\Lambda]}{\Tr[\Pi_\Lambda]} \\
        & = \frac{\one_A}{d^{|A|}}(1 + O(e^{-|B|/\xi})),
    \end{align}
    Hence,
    \begin{align}
        \norm{\Tr_{B}[\rho_{\beta = 0, \Lambda}] - \one_A/d^{|A|}}_1 = O\left(\norm{\one_A/d^{|A|}}_1e^{-|B|/\xi}\right) = O(e^{-|B|/\xi}).
    \end{align}
\end{proof}

Finally, we end with a theorem that a similar result for local indistinguishability of canonical and Gibbs ensembles holds true even at finite temperatures,   

\begin{theorem}
    For $G$ a finite Abelian group and $H = \sum_{i=1}^N h_i$ a local Hamiltonian with $\norm{h_i} \leq J$ and $|\supp h_i| \leq k$, then
    \begin{equation}
        \norm{\Tr_{A^C}[\rho_{\beta,\Lambda}] - \Tr_{A^C}[\rho_\beta]}_1 = O(e^{-(c_\beta N-|A|)/\xi_\beta}),
    \end{equation}
    where $\xi_\beta, c_\beta > 0$ if $\beta < \beta_*$, with $\xi_\beta$, $c_\beta$ and $\beta_*$ independent of system size $N$.
\end{theorem}
\begin{proof}
    We first rewrite the difference of states above in terms of a difference in partition functions and a difference of Gibbs exponentials. For that, let $B = A^C$, and $Z_\beta \defeq \Tr[e^{-\beta H}]$, $Z_{\beta,\Lambda} = \Tr[e^{-\beta H} \Pi_\Lambda]$:
    \begin{align}
        \norm{\Tr_{B}[\rho_{\beta,\Lambda}] - \Tr_{A^C}[\rho_\beta]}_1 & = \norm{-\left(\frac{|G/K| Z_{\beta,\Lambda}}{Z_{\beta}} - 1\right) \frac{\Tr_B[e^{-\beta H} \Pi_\Lambda]}{Z_{\beta,\Lambda}} + \frac{1}{Z_{\beta}} (|G/K| \Tr_B[e^{-\beta H} \Pi_\Lambda] - \Tr_B[e^{-\beta H}])}_1 \\ 
        & \leq Z_\beta^{-1} ||G/K| Z_{\beta, \Lambda} - Z_\beta|\norm{\Tr_{B}[\rho_{\beta, \Lambda}]}_1 + Z_{\beta}^{-1} \norm{|G/K| \Tr_B[e^{-\beta H} \Pi_\Lambda] - \Tr_B[e^{-\beta H}]}_1. \label{eq:two_terms_bound}
    \end{align}
    The first term above is nothing but the second term specialized to $A = \emptyset$. Hence, we will focus on bounding the latter. Without loss of generality, we can assume $H \leq 0$ by adding to it a constant term proportional to the identity operator that will not change $\rho_\beta$ nor $\rho_{\beta,\Lambda}$. In that case, $Z_{\beta} = \Tr[e^{-\beta H}] \geq d^N$. By applying a high-temperature expansion of the exponentials, we arrive at
    \begin{align}
        Z_\beta^{-1} \lVert|G/K|\Tr_B[e^{-\beta H }\Pi_{\Lambda}] - \Tr_B[e^{-\beta H}] \rVert_1 & \leq d^{-N} \sum_{n=0}^{\infty} \frac{\beta^n}{n!} \lVert|G/K|\Tr_B[H^n\Pi_{\Lambda}] - \Tr_B[H^n] \rVert_1. \label{eq:taylor_series_sum}
    \end{align}
    We now focus on each term of the sum above,
    \begin{align}
        \lVert|G/K|&\Tr_B[H^n \Pi_{\Lambda}] - \Tr_B[H ^n] \rVert_1 \nonumber \\ 
        &\leq \sum_{i_1, \ldots, i_n} \norm{ |G/K| \Tr_B\left[\prod_{\ell=1}^n h_{i_\ell} \Pi_\Lambda\right] - \Tr_B\left[\prod_{\ell=1}^n h_{i_\ell}\right] }_1 \label{eq:bound_Second_term}\\
    \end{align}
    At this stage, it is useful to introduce a notation for the region with non-trivial action of $\prod_{\ell = 1}^{n}h_{i_{\ell}}$, $S_{\{i_\ell\}} = \cup_{\ell=1}^n\supp h_{i_\ell}$. We decompose the trace over $B$ into trace over $B\cap S_{\{i_{\ell}\}}$ and $B\cap S_{\{i_{\ell}\}}^{C}$. We also decompose $\Pi_\Lambda = \sum_\lambda \Pi_\lambda^R \Pi_{\Lambda \overline{\lambda}}^{R^C}$ for $R = A \cup S_{\{i_\ell\}}$, such that $R^C \equiv B \cap S_{\{i_\ell\}}^C$. We also use the fact that $\sum_{\lambda}\Pi_{\lambda}^{R} = \one_{R}$ to simplify the second term in Eq.~\eqref{eq:bound_Second_term}: 
    \begin{align}
        \lVert|G/K|&\Tr_B[H^n \Pi_{\Lambda}] - \Tr_B[H ^n] \rVert_1 \nonumber \\
        &\leq \sum_{i_1, \ldots, i_n}\norm{\sum_{\lambda}\Tr_{B\cap S_{\{i_{\ell}\}}}\left[\prod_{\ell = 1}^{n}h_{i_{\ell}}\Pi_{\lambda}^{A\cup S_{\{i_\ell\}}}\right]\left(|G/K|\Tr_{B\cap S_{\{i_\ell\}}^C}\left[\Pi_{\Lambda \overline{\lambda}}^{B \cap S_{\{i_\ell\}}^C}\right]-\Tr_{B \cap S_{\{i_\ell\}}^C}\left[\mathbb{I}_{B \cap S_{\{i_\ell\}}^C}\right]\right)}_1 \\
        & = \sum_{i_1, \ldots, i_n} \sum_{\lambda} \norm{\Tr_{B\cap S_{\{i_{\ell}\}}}\left[\prod_{\ell=1}^n h_{i_\ell} \Pi^{A \cup S_{\{i_{\ell}\}}}_\lambda\right]}_1 \left| |G/K| \Tr_{B\cap S_{\{i_{\ell}\}}^C}\left[\Pi_{\Lambda \overline{\lambda}}^{B\cap S_{\{i_{\ell}\}}^C}\right] - d^{|B\cap S_{\{i_{\ell}\}}^C|} \right| \\
        & \leq C \sum_{i_1, \ldots, i_n} \sum_{\lambda} \norm{\Tr_{B\cap S_{\{i_{\ell}\}}}\left[\prod_{\ell=1}^n h_{i_\ell} \Pi^{A \cup S_{\{i_{\ell}\}}}_\lambda\right]}_1 r^{|B\cap S_{\{i_{\ell}\}}^C|}~\text{ (using Lemma \ref{lemma:trace_irrep_projector}.)} \\
        & \leq C \sum_{i_1, \ldots, i_n} \sum_\lambda \norm{\prod_{\ell=1}^n h_{i_\ell} \Pi^{A \cup S_{\{i_{\ell}\}}}_\lambda}_1 r^{|B|-\max(|B|,kn)} \\
        & = C \sum_{i_1, \ldots, i_n} \norm{\prod_{\ell=1}^n h_{i_\ell}}_1 r^{|B|-\max(|B|,kn)} \\
        & \leq C \sum_{i_1, \ldots, i_n} d^{|A|+\max(|B|,kn)} J^n r^{|B|-\max(|B|,kn)} \\
        & \leq C  
        \begin{cases}
        r^N (d/r)^{|A|} (J N d^k / r^k)^n, & \text{if } kn < |B| \\
        d^N (JN)^n & \text{if } kn \geq |B|
        \end{cases}
    \end{align}
    where in the third line we have used Lemma \ref{lemma:trace_irrep_projector} in the form $||G/K|  \Tr[\Pi^R_\lambda] - d^{|R|}| < C r^{|R|}$, with $C$ a constant and, crucially, $r < d$; and in the fifth line we used that $\prod_{\ell=1}^n h_{i_\ell}$ is block diagonal with respect to the subspaces projected by $\Pi_{\lambda}^{A \cup S_{\{i_\ell\}}}$ because the product commutes with the symmetry acting on $A \cup S_{\{i_\ell\}}$. Putting it back in the original sum (Eq. \eqref{eq:taylor_series_sum}),
    \begin{align}
        Z_\beta^{-1} \lVert|G/K|\Tr_B[e^{-\beta H }\Pi_{\Lambda}] &- \Tr_B[e^{-\beta H}] \rVert_1 \\ 
        & \leq C \left( (r/d)^{N-|A|} \sum_{n=0}^{\lceil |B|/k \rceil -1} \frac{(\beta J N d^k/r^k)^n}{n!} + \sum_{n=\lceil |B|/k \rceil}^\infty \frac{(\beta JN)^n}{n!} \right) \\
        & \leq C \left( (r/d)^{N-|A|} \sum_{n=0}^{\infty} \frac{(\beta JN d^k/r^k)^n}{n!} + e^{\beta J N} \frac{( \beta JN)^{\lceil |B|/k \rceil}}{\lceil |B|/k \rceil!}  \right) \\ 
        & \leq C \left((r/d)^{N-|A|} e^{\beta J N d^k/r^k}+ e^{\beta J N } \left(\frac{e \beta J N}{\lceil |B|/k \rceil}\right)^{\lceil |B|/k \rceil} \right) \\
        & \leq C (e^{-(\ln(d/r) - \beta J d^k/r^k)N + \ln(d/r) |A|} + e^{\beta J N} (e \beta J k)^{N - |A|}  (1 - |A|/N)^{-(N-|A|)}) \\
        & \leq C (e^{-(\ln(d/r) - \beta J d^k/r^k)N + \ln(d/r) |A|} + e^{(\beta J + \ln(e \beta J k))N  + (-\ln(e \beta  J k) + 1)|A|}) \\
        & = O(e^{-(c_\beta N - |A|)/\xi_\beta}),
    \end{align}
    where the values of $\xi_\beta$ and $c_\beta$ depend if $r = 0$ or not: if $r=0$, then the first term above vanishes, and so $\xi_\beta^{-1} = 1 - \ln(e \beta Jk)$ and $c_\beta = -\xi_\beta (\beta J + \ln(e \beta J k))$; and if $r > 0$, then, at small enough $\beta$, the first term dominates the second, and we have $\xi_\beta^{-1} = \ln(d/r)$ and $c_\beta = 1 - \xi_\beta \beta J d^k/r^k$. Importantly, in both cases we have $\xi_\beta, c_\beta > 0$ if $\beta < \beta_*$ for $\beta_*$ small enough.

    Since this bound is tighter the smaller region $A$ is, then the first term of Eq. \eqref{eq:two_terms_bound} is smaller than the second one, and we arrive at the conclusion.
\end{proof}

\end{document}